\documentclass[journal]{IEEEtran}

\ifCLASSINFOpdf
\else
   \usepackage[dvips]{graphicx}
\fi
\usepackage{url}
\usepackage{caption}
\usepackage{graphicx}
\usepackage{cite}
\usepackage{amsmath,amssymb,amsfonts}
\usepackage{graphicx}
\usepackage{textcomp}
\usepackage{xcolor}
\usepackage{MnSymbol}
\usepackage{subcaption}
\usepackage{cuted}
\usepackage{amsthm}
\newtheorem{proposition}{Proposition}
\usepackage[linesnumbered,ruled,vlined]{algorithm2e}
\begin{document}

\title{ \color{black}Precoding for Multi-Cell ISAC: from Coordinated Beamforming to Coordinated Multipoint and Bi-Static Sensing\\}

\author{Nithin Babu, ~\IEEEmembership{~Member, ~IEEE}, Christos Masouros, ~\IEEEmembership{~Fellow, ~IEEE}, \\Constantinos B. Papadias,~\IEEEmembership{~Fellow,~IEEE}, \color{black}and Yonina C. Eldar, ~\IEEEmembership{~Fellow,~IEEE}. 
\thanks{This version of the paper has been submitted to IEEE Transactions on Wireless Communications for possible publication. This work was supported in part by the Engineering and Physical Sciences Research Council under Project EP/S028455/1.}
\thanks{N. Babu and C. Masouros are with the Department of EE, UCL (e-mail: n.babu@ucl.ac.uk, c.masouros@ucl.ac.uk).}
\thanks{C. B. Papadias is with Research, Technology and Innovation Network (RTIN), Alba, 
The American College of Greece, Greece (e-mail: cpapadias@acg.edu).}
\thanks{Y. C. Eldar is with the Faculty of Mathematics and ComputerScience, Weizmann Institute of Science, Rehovot, Israel (e-mail:   yonina.eldar@weizmann.ac.il.)}
}

\markboth{ }
{Shell \MakeLowercase{\mathrmit{et al.}}: Bare Demo of IEEEtran.cls for IEEE Journals}
\maketitle

\begin{abstract}
This paper proposes a framework for designing robust precoders for a multi-input single-output (MISO) system that performs integrated sensing and communication (ISAC) across multiple cells and users. We use  Cramer-Rao-Bound (CRB) to measure the sensing performance and derive its expressions for two multi-cell scenarios, namely coordinated beamforming (CBF) and coordinated multi-point (CoMP). \color{black}In the CBF scheme, a BS shares channel state information (CSI) and estimates target parameters using monostatic sensing. In contrast, a BS in the CoMP scheme shares the CSI and data, allowing bistatic sensing through inter-cell reflection\color{black}. We consider both \color{black} block-level (BL) and symbol-level (SL) \color{black}precoding schemes for both the multi-cell scenarios that are \color{black}robust to channel state estimation errors\color{black}. The formulated optimization problems to minimize the CRB in estimating the parameters of a target and maximize the minimum communication signal-to-interference-plus-noise-ratio (SINR) while satisfying a given total transmit power budget are non-convex. We tackle the non-convexity using a combination of semidefinite relaxation (SDR) and alternating optimization (AO) techniques. Simulations suggest that neglecting the inter-cell reflection and communication links degrades the performance of an ISAC system. The CoMP scenario employing \color{black}SL precoding \color{black} performs the best, whereas the \color{black}BL precoding \color{black} applied in the CBF scenario produces relatively high estimation error for a given minimum SINR value.  
\end{abstract}
\IEEEpeerreviewmaketitle
\section{Introduction}
\color{black}Integrated sensing and communication (ISAC) has been identified as an enabler for next-generation wireless networks to augment communication services with \color{black}sensing for emerging \color{black} applications spanning \color{black}connected vehicles, remote healthcare, smart homes, and more \cite{liu2020joint},\cite{liu2022integrated}. The availability of large bandwidth, multiple antennas, and dense deployment of 5G-Advanced and 6G networks enable high-resolution radio sensing capability \cite{chen20235g}. ISAC has been recognized as one of the key enablers for 6G by the standardization bodies such as the International Telecommunication Union (ITU) and Third Generation Partnership Project (3GPP) \cite{union2022future} \cite{kaushik2023towards}. Having sensing and communication capabilities mutually benefits both systems: sensing can be used by the communication system to understand the environment better and enhance, for instance, interference management and beamforming, whereas the connected network infrastructure enables coordinated sensing at an unprecedented scale. 

%The communication users in the network decode the transmitted symbol from the received symbol using the knowledge of the channel state information (CSI), whereas the radar receiver at a base station (BS) uses the knowledge of the transmitted data to estimate the radar channel and, hence targets' parameters. Hence, the transmitted symbol vector is a common input determining sensing and communication performance. 
\color{black}Numerous works have considered designing optimal transmit waveforms/precoders to improve a single-cell ISAC system's sensing and communication performance\cite{9858656,liu2018toward,5776640,9723383,temiz2023experimental,liu2021cramer,9124713,9303435,7953658,9385108,xu2022proof,ozkaptan2023software,liao2023faster,9585492,10065807}. %The two main approaches are orthogonal resource allocation \cite{9858656} and unified waveform design \cite{liu2018toward}. %The former allocates orthogonal time or frequency resources to avoid interference between the sensing and the communication functions. Such an approach is easy to implement with minimal change to the existing infrastructure but is less efficient. 
However, the dense deployment of next-generation small-cell BSs causes the signals transmitted from a BS to its users to affect the sensing and communication performance of neighboring BSs. When the BSs use the same time and frequency resources to serve users, a user experiences intra-cell interference by the signals intended for other users in the same cell and inter-cell interference (ICI) due to co-channel signal leakage from the neighboring BSs, reducing the received SINR. The work in \cite{bjornson2013optimal} discusses various resource allocation schemes and precoder designs to improve the received SINR for a communication-only multi-cell system by suppressing interference. However, for ISAC, in addition to this, inter-cell reflections (ICR) will be received by a BS from its target due to the signal transmitted from the neighboring BSs. The received power through ICR can degrade target parameter estimation if the BS is unaware of the data transmitted from the neighboring BS. \color{black}This emphasizes the need for coordination among BSs to effectively manage inter-cell communication and sensing links. %Hence, we consider inter-cell communication and sensing links while designing precoders.

The ISAC BSs can coordinate at different levels through coordinated beamforming (CBF) and coordinated multipoint (CoMP) schemes, similar to the communication-only multi-cell system. \color{black}In an ISAC CBF scheme, a single BS serves a disjoint set of users and estimates its target's parameters through mono-static sensing\color{black}. Here, the signal power received through the inter-cell links negatively affects the sensing and communication performances. Hence, each BS selects a transmit strategy jointly with other BSs to minimize the ICI and ICR using the globally shared CSI obtained from the users via feedback channels and ICR direction information shared among neighboring BSs. Conversely, a user is served by all the BSs in an ISAC CoMP scenario, thereby improving the received communication SINR at the user. Since the CoMP scheme shares the user data to be transmitted and the CSI globally among the BSs through a high bandwidth backhaul network, the additional power received through ICR and ICI links enhances estimation accuracy and communication SINR at the expense of increased coordination overhead. \color{black}In particular, ICR enables bistatic sensing\color{black}. 

Each of these multi-cell scenarios can employ either of two existing precoding strategies that differ in how they deal with the co-channel interference experienced by a user: a) block-level precoding (BLP)\cite{bjornson2013optimal} and b) symbol-level precoding (SLP)\cite{masouros2015exploiting}. %We consider applying both precoding strategies %to design optimal precoders that jointly minimize the target's parameter estimation error and maximize the minimum communication SINR for BSs coordinating 
%in the CBF and CoMP modes. 
Conventional precoding, also called BLP, designs precoders for a set of users to transmit a given block of symbols. It treats the interference experienced by a user as a harmful element, whereas constructive interference (CI)-based precoding, also called SLP, uses the signal power received through the instantaneous interference to aid the received communication SINR. The SLP technique exploits knowledge of both CSI and downlink data to be transmitted at the BS to ensure the received symbol at a user falls within the constructive region of the signal constellation \cite{li2020tutorial}. \color{black}Considering the SLP in the CoMP scheme, each BS can utilize both the intra-cell and inter-cell interference to enhance the received signal power\color{black}. Conversely, the SLP in the CBF scheme can only utilize the intra-cell interference due to the non-availability of the user data from the neighboring BSs. \color{black}The design of the optimal precoders for both coordination schemes depends on the available CSI at the BSs and the users. In practice, CSI is prone to errors necessitating that the designed precoders exhibit robustness to potential CSI uncertainties.\color{black}

\color{black} Multi-cell ISAC setups with various combinations of CBF and CoMP schemes have been considered in \cite{wei2023symbol,wang2023resource,li2022beamforming,xu2023fundamental,xu2022integrated,xu2023integrated,xu2023joint,jiang2023collaborative,liu2023distributed,zhang2023joint}. The work in \cite{wei2023symbol} and \cite{wang2023resource} consider orthogonal transmission among BSs, whereas \cite{li2022beamforming} uses an additional BS as the receiver to enable bi-static sensing and hence does not consider the inter-cell interference links. %\cite{xu2022integrated} and \cite{xu2023fundamental} investigate a cell-free architecture; hence, the proposed signal models do not apply to the CBF and CoMP schemes. 
%In \cite{xu2022integrated}, %individual remote radio heads (RRHs) sense distinct targets by transmitting separate sensing signals, while all RRHs jointly serve the users.  As a result, 
%the sensing model incorporates interference from other targets while the users are served using the CoMP scheme. 
\cite{xu2023fundamental} considers a multi-static scenario where one  BS  transmits signals  to  one  vehicle and the echoes are captured by several other BSs for sensing purposes. \cite{xu2022integrated} and \cite{xu2023integrated} consider ComP to serve a set of users, while the sensing model considers reflection from unwanted targets as interference. The study in \cite{xu2023joint} considers CoMP mode to serve a set of users while detecting a target. Here, the transmitter employs distinct sensing and communication signals, allowing for modeling the interference between them.  The authors of \cite{jiang2023collaborative} model multi-cell ISAC interference considering 3 BSs operating in CBF mode where the users are considered as targets. \cite{liu2023distributed} models the inter-user and inter-subsystem interference. In \cite{zhang2023joint}, the BSs serve the users in the CBF mode and neglect the inter-cell reflection links. \color{black}
   %The authors of  \cite{liao2023faster}, \cite{9585492}, and \cite{10065807} considered using SLP-based ISAC system design in single-cell scenarios. The work in \cite{liao2023faster} ensures communication per-user quality-of-service while minimizing the minimum mean squared error (MMSE) for the sensed parameter estimation through CI methodologies. \cite{9585492} considers the weighted minimization of the radar waveform similarity error and the CI-based communication waveform design, whereas \cite{10065807} minimizes the CRB subject to a constant-modulus, similarity, and the communication QoS constraints. Nevertheless, these works consider a single-cell ISAC setup. %Recent advances in ISAC systems are well summarized in \cite{lu2023integrated}.
% \subsection{Related Works}
%\color{black}
%Researchers have investigated various multi-cell ISAC setups in \cite{wild2021joint,wei2023symbol,xu2022integrated,li2022beamforming,xu2023integrated,xu2023joint,xu2023fundamental,jiang2023collaborative,liu2023distributed,wang2023resource,zhang2023joint}. 
The work in \cite{xu2022integrated,li2022beamforming,xu2023joint,xu2023integrated} consider BLP design to achieve various objectives spanning minimizing energy consumption \cite{xu2022integrated}, beampattern mismatch error \cite{xu2023integrated}  or total transmit power \cite{xu2023joint} to maximizing the minimum communication SINR \cite{li2022beamforming} subject to radar and or communication constraints. %Furthermore, \cite{xu2023joint} and \cite{xu2023fundamental} propose BS selection algorithm considering the tradeoff between the location accuracy and system complexity. 
\cite{jiang2023collaborative} employs a collaborative SLP to mitigate mutual interference in a 3-cell ISAC system and minimize the total transmit power.%, whereas \cite{liu2023distributed} uses an unsupervised learning approach to manage the interference and maximize the expected capacity of the system. Optimal sub-band allocation, user association, and power allocation schemes that maximize the sum rate and network utility functions are proposed in \cite{wang2023resource} and \cite{zhang2023joint}.  

\subsection{Main Contributions and Paper Organization}
\color{black}The existing multi-cell ISAC works consider no uncertainty in CSI. Moreover, the work in \cite{wei2023symbol,wang2023resource,li2022beamforming,xu2023fundamental,xu2022integrated,xu2023integrated,xu2023joint,jiang2023collaborative,liu2023distributed,zhang2023joint} employ either orthogonal schemes or CoMP for serving the users, and the sensing model follows the properties of CBF in which the reflections from other targets are considered interference. %This is justified by the assumption of using separate signals for communication and sensing. 
The proposed sensing models do not capture the property that the reflections from other targets can aid or degrade the sensing depending on the coordination level among the BSs. Furthermore, the sensing performance is quantified using radar SINR value. Explicit optimization of estimation performance metrics has not been considered in \cite{wei2023symbol,wang2023resource,li2022beamforming,xu2023fundamental,xu2022integrated,xu2023integrated,xu2023joint,jiang2023collaborative,liu2023distributed,zhang2023joint}.  \color{black}We considered this aspect in our previous work \cite{babu2023multi} while deriving the CRB expression for estimating the azimuth angle of a target in a multi-cell ISAC setup. Subsequently, this expression is utilized to design block-level precoders that optimize a weighted combination of sensing and communication metrics without considering channel state uncertainty. It is important to note that the derived CRB expression becomes invalid when dealing with scenarios involving more than one unknown parameter per target or multiple targets per cell. Recall that the precoding design should be robust to CSI uncertainties. %Numerous works explain how to model the SINR in both the CBF and CoMP scenarios \cite{bjornson2013optimal}, \cite{wei2019multi}; however, measuring the sensing performance requires a performance metric that can adapted to the scenario under consideration. CRB acts as the lower bound for the estimation error variance; however, the CRB expression used in \cite{liu2021cramer} only applies to a single-cell system demanding a new CRB expression that can be applied to the CBF and CoMP scenarios. Moreover, the precoding design relies on the available CSI information at the BS end; hence, the obtained solution should be robust to CSI uncertainties. 
The work in \cite{wei2019multi} designed robust symbol-level precoders for a multi-cell communication system that minimizes the total transmission power without considering the sensing capability. Even though the problem formulations in the aforementioned works consider multi-cell ISAC scenarios, a relative system performance analysis under different levels of coordination amongst BSs spanning from CBF to CoMP in the presence of channel state estimation error has not been investigated. To the best of our knowledge, explicit optimization of a weighted combination of the estimation error and communication performance metrics under CSI uncertainty has not been considered in the context of a multi-cell ISAC network. \color{black} Our main contributions are: 
\begin{itemize}
 \item We propose a robust precoder design framework to minimize the target parameter estimation error variance and maximize the minimum communication SINR experienced by a user in a multi-cell multi-user MIMO ISAC system \color{black}under CSI uncertainty\color{black}.
    
    \item We consider two multi-cell scenarios, CBF and CoMP, and \color{black}extend the corresponding CRB expressions in \cite{babu2023multi} to multiple parameters (location and complex amplitudes) per target. \color{black}The important distinction between the two cooperation modes in terms of sensing is that in the CRB, the adjacent BSs' signals act as interference to the monostatic sensing of the serving BS, while in CoMP, the same signals can be used as an additional source of bi-static sensing\color{black}.
    \item The derived \color{black}CRB \color{black} expressions are then utilized to formulate optimization problems that jointly minimize the CRB value and maximize the minimum communication SINR value subject to a total transmit power budget. The formulations consider block-level and symbol-level precoding techniques for the CBF and CoMP scenarios.   
    \item We solve all the non-convex optimization problems using a combination of the semidefinite relaxation (SDR) and alternating optimization (AO) methods.
    \item  Finally, we compare the sensing and communication performances of all the considered precoder design schemes through simulations.    
\end{itemize}
The remaining content of the paper is organized as follows: in Section \ref{systemmodel}, we explain the system model and derive the CRB expressions for the CBF and CoMP multi-cell scenarios. Section \ref{BLP} and Section \ref{SLP} discuss the block-level and symbol-level robust precoding schemes for the considered multi-cell scenarios. Finally, in Section \ref{numerical}, we present our main findings through numerical evaluation, and the paper is concluded in Section \ref{conclusion}.
\subsubsection*{Notations} Matrices, vectors, and scalars are denoted by bold uppercase, lowercase, and normal font letters, respectively. We use $\mathrm{tr}()$, $()^{\mathrm{T}}$, $()^{\mathrm{H}}$, and $()^{{*}}$ to represent trace operation, transpose,  Hermitian transpose,  and the complex conjugate of the matrices or vectors. The real and imaginary parts of $\mathbf{x}$ are represented as $\mathbf{x}^{\mathrm{R}}$ and $\mathbf{x}^{\mathrm{I}}$, respectively, and $\|\|$ denotes the $l_{2}$ norm. We represent an $N \times N$ null matrix and an $N \times N$ identity matrix as $\mathbf{0}_{N}$ and $\mathbf{I}_{N}$, respectively. Additionally, $\Vec{0}_{N}$ denotes an $N \times 1$ null vector.
\section{System Model} \label{systemmodel}
\begin{figure}[]
\centering
\captionsetup{justification=centering}
\centerline{\includegraphics[width=0.95\columnwidth]{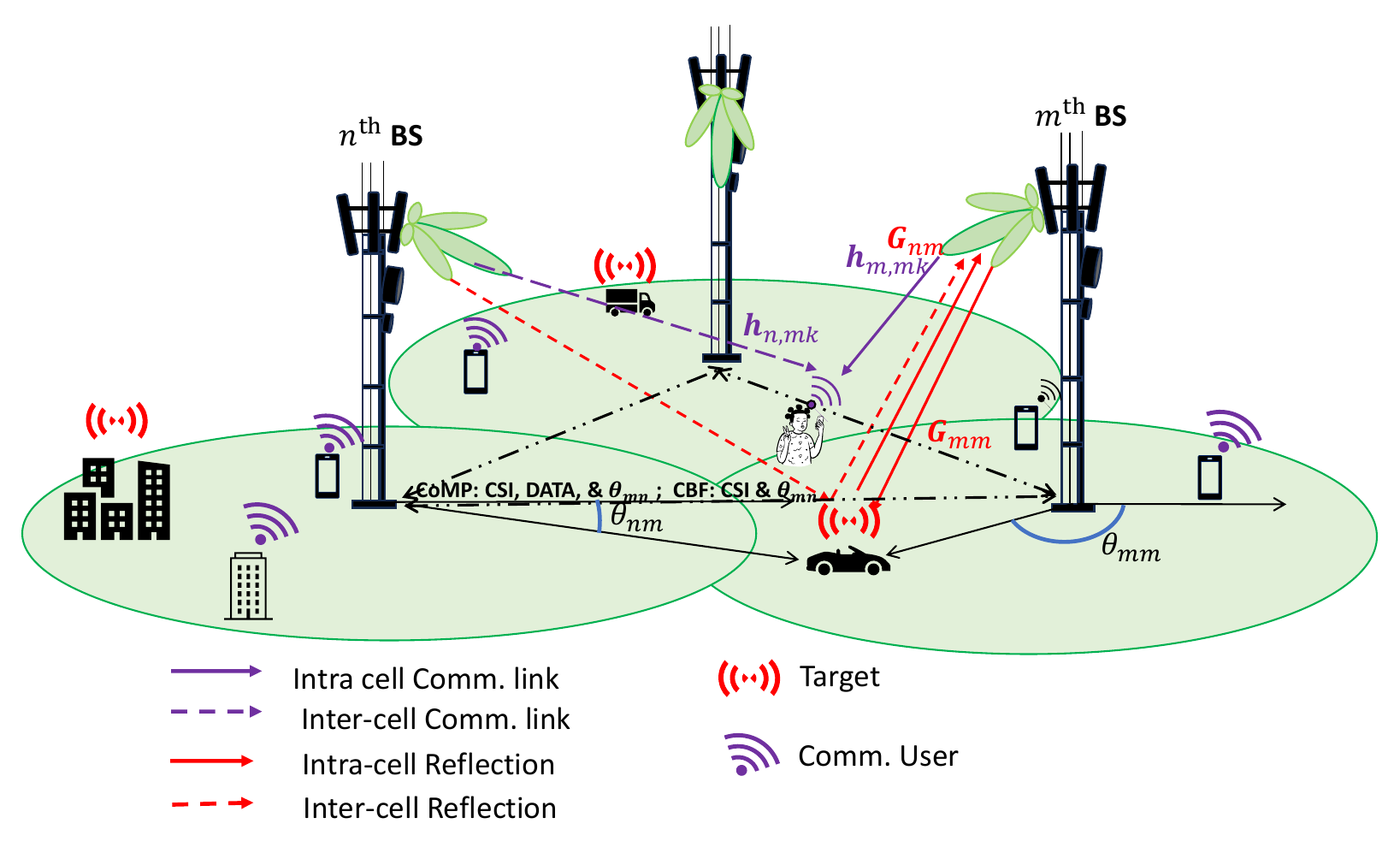}}
\caption{\color{black}System setup.}
\label{figure1}
\end{figure}
% We consider a multi-cell, multi-user MISO system during downlink transmission. The system consists of $J$ cells, each containing a target, $K$ users and a base station (BS) equipped with $N_{\mathrm{t}}$ transmit and $N_{\mathrm{r}}$ receive antennas. Let $\mathbf{x}_{m}=\sum_{n=1}^{K}\mathbf{w}_{m,n}s_{m,n}\in \mathbb{C}^{N_{\mathrm{t}}\times 1}$ be a narrowband signal vector transmitted by the $m^{\mathrm{th}}$ BS; $\mathbf{w}_{m,n}\in \mathbb{C}^{N_{\mathrm{t}}\times 1}$ and $s_{m,n}$ denote the precoding and transmitted data to $n^{\mathrm{th}}$ user covered by the $m^{\mathrm{th}}$ BS: $U_{m,n}$. Since all the BSs transmit the respective signal vector, $\{\mathbf{x}_{m}\}$, simultaneously, the received reflected echo signal at the $m^{\mathrm{th}}$ BS from the target in the cell is given by,
%\color{blue}To enable low-complexity and energy-efficient receivers, we assume single-user detection meaning that a user is not attempting to decode and subtract interfering signals while decoding its own signals. When the radar has no a priori knowledge about targets,
% %the initial step is to search for potential targets in the whole
% space. Similarly, when no channel information is available
% at the communication system, the CSI has to be estimated
% before any useful information can be decoded at the receiver. More specifically, in our case, the BS first sends omnidirectional DL pilots (DP), and then estimates the K AoAs in  as well as the associated range and Doppler
% parameters of all K targets. 
\color{black}We consider a multi-cell multi-input single-output (MISO) ISAC system with $J$ cells and $K$ single-antenna users per cell. In addition to the users,  each cell has a target and a BS with a uniform linear array (ULA) of $N_{\mathrm{tx}}$ transmit antennas spaced at $\lambda/2$ distance, where $\lambda$ is the wavelength. The BS is also equipped with a $\lambda/2$-spaced receive ULA of $N_{\mathrm{rx}}$ elements for sensing and isolated from the transmit antenna elements. Moreover, the BSs are interconnected through optical cables to synchronize and exchange information. \color{black}Furthermore, one of the BSs acts as the central node where the precoder design is assumed to occur. The obtained solutions are then communicated to the respective BSs through the optical cables\color{black}.
\subsection{Communication Model and Performance Metric}
\color{black}For the communication service\color{black}, we aim to maximize the minimum SINR value experienced by a user in a cell. Recall that the corresponding SINR expressions vary depending on the BSs' coordination level (CBF or CoMP) and precoding mode, BLP or SLP, and are detailed in the following. The proposed algorithm requires the channel state information (CSI) to be available at the BS and the users. We consider the system operating in time division duplexing (TDD) mode so that the downlink CSI can be derived from the uplink channel observations based on the received uplink sounding reference signal (SRS) transmitted by the users. Additionally, each user employs a simple equalizer for the composite channel $\mathbf{h}^{T}\mathbf{w}$, where $\mathbf{h}$ is the communication channel from the serving BS and $\mathbf{w}$ is the precoding vector. 
\subsubsection{Coordinated Beamforming}
Let $U_{mk}$ represent the $k^{\mathrm{th}}$ user of the $m^{\mathrm{th}}$ cell. In practice, the CSI is imperfectly known at the BS. Let the CSI uncertainty be limited within the spherical set $\mathcal{E}_{mk}=\{\mathbf{e}_{mk}:\|\mathbf{e}_{mk}\|^{2}\leq \delta^{2}\}$. Then, the actual CSI, $\mathbf{\tilde{h}}_{n,mk}$, from the $n^{\mathrm{th}}$ BS to $U_{mk}$ is the sum of the observed CSI, denoted as $\{\mathbf{h}_{n,mk}\}$, and $\mathbf{e}_{mk}$:
\begin{align}
   \mathbf{\tilde{h}}_{n,mk} &= \mathbf{{h}}_{n,mk}+\mathbf{e}_{mk}, \, \forall U_{mk}.
\end{align}
Let the $m^{\mathrm{th}}$ BS transmits the symbol matrix $\mathbf{X}_{m}=\left[\mathbf{x}_{m1},\mathbf{x}_{m2},...,\mathbf{x}_{mL}\right]=\mathbf{W}_{m}\mathbf{S}_{m}  \, \in \mathcal{C}^{N_{\mathrm{tx}}\times L}$, where $\mathbf{W}_{m}=\left[\mathbf{w}_{m1},\mathbf{w}_{m2},...,\mathbf{w}_{mK}\right]$ $\forall$ $m \in \{1,2,.., J\}$ are the dual-functional beamforming matrices to be designed, with $L > N_{\mathrm{tx}} $ being the length of the radar pulse/ communication frame. Here, $\mathbf{S}_{m}\in \mathcal{C}^{K \times L}$ is the orthogonal data stream transmitted to $K$ users of the $m^{\mathrm{th}}$ BS: $(1/L)\mathbf{S}_{m}\mathbf{S}^{H}_{m}=\mathbf{I}_{K}$. Then the received signal at $U_{mk}$ is expressed as 
\begin{IEEEeqnarray}{rCl}
\mathbf{y}_{mk}^{\mathrm{C}} &=& \mathbf{\tilde{h}}^{T}_{m,mk}\mathbf{X}_{m}+\sum_{n\neq m}^{J} \mathbf{\tilde{h}}^{T}_{n,mk}\mathbf{X}_{n}+ \mathbf{z}_{mk}^{\mathrm{C}},\label{yc} 
\end{IEEEeqnarray}
where $\mathbf{z}_{mk}^{\mathrm{C}} \, \in \mathcal{C}^{1 \times L}$ is an AWGN noise vector with variance of each entry being $\sigma^{2}_{\mathrm{C}}$. Since no data is shared among the BSs in CBF mode, the first term in the RHS of \eqref{yc} contains the useful signal and intra-cell interference from the users of the same cell, whereas the second term represents the inter-cell interference from the neighboring cells. In a multi-carrier system, for example, based on orthogonal frequency-division multiplexing (OFDM), the input–output model \eqref{yc} could describe one of the subcarriers. 
\subsubsection{Coordinated Multipoint}
Since the BSs share the user data and the CSI in the CoMP mode, the $J\cdot K$ users in the multi-cell ISAC system can be considered to be served by a virtual single-cell BS with $N=J\cdot N_{\mathrm{tx}}$ antennas. In this scenario, the precoder design occurs at a designated base station serving as the central node. The resulting precoder matrices are then shared with the respective base stations through optical cables. Let $U_{k}$ for $k \in \{1,2,..,J\cdot K\}$ represent the $k^{\mathrm{th}}$ user of the virtual single-cell system. Note that $U_{k}$ for $k \in \{K\cdot (m-1)+1,K\cdot i+2,…,K\cdot m\}$ represents $K$ users the $m^{\mathrm{th}}$ BS. Let $\mathbf{\tilde{h}}_{k}\in \mathcal{C}^{N\times 1}$ be the actual channel vector from all the BSs to $U_{k}$ represented as
\begin{align}
   \mathbf{\tilde{h}}_{k} &= \mathbf{{h}}_{k}+\mathbf{e}_{k},
\end{align}
where $\mathbf{e}_{k} \in \mathcal{E}_{k}=\{\mathbf{e}_{k}:\|\mathbf{e}_{k}\|^{2}\leq J\delta^{2}\}$. The received signal at $U_{k}$ is expressed as 
\begin{IEEEeqnarray}{rCl}
\mathbf{y}_{k}^{\mathrm{C}} &=& \mathbf{\tilde{h}}^{T}_{k}\mathbf{X}+ \mathbf{z}_{k}^{\mathrm{C}}, \label{yccomp}
\end{IEEEeqnarray}
where $\mathbf{X}=\left[\mathbf{X}_{1};\mathbf{X}_{2},...;\mathbf{X}_{J}\right]\in \mathcal{C}^{N \times L}$ is the concatenated symbol matrix available at each BS and $\mathbf{z}_{k}^{\mathrm{C}} \, \in \mathcal{C}^{1 \times L}$ is an AWGN noise vector with variance of each entry being $\sigma^{2}_{\mathrm{C}}$.
Since the intra-cell and inter-cell interference are differently treated by BLP and SLP, we will give the corresponding SINR expressions in the respective sections.\color{black}
\subsection{Sensing Model and Performance Metric}
\color{black}As shown in Fig. \ref{figure1}, the BSs transmit simultaneously, and each BS receives its echo signal and multiple echo signals from the neighboring BSs due to ICR; we consider the dominant path among the ICR paths and the signal power received through the weaker paths is included in the noise term. Hence, the resulting echo signal received by the $m^{\mathrm{th}}$ BS from its target is given as
\begin{IEEEeqnarray}{rCl}
\mathbf{Y}_{m}^{\mathrm{I}} &=& \mathbf{G}_{mm}\mathbf{X}_{m}+ \underbrace{\sum_{n\neq m}^{J} \mathbf{G}_{nm}\mathbf{X}_{n}}_{\mathrm{ICR}}+ \mathbf{Z}_{m}^{\mathrm{R}},\label{yr} 
%&=& \mathbf{G}_{m} \mathbf{x} + \mathbf{z}_{m}^{\mathrm{R}}.
\end{IEEEeqnarray}
where %$\mathbf{G}_{m}=[\mathbf{G}_{1,m}, \mathbf{G}_{2,m},...,\mathbf{G}_{J,m}] \in \mathbb{C}^{N_{\mathrm{r}}\times N_{\mathrm{t}}\cdot J }$ and $\mathbf{x} \in \mathbb{C}^{N_{\mathrm{t}}\cdot J\times 1 }$ .
$\mathbf{G}_{nm}=\alpha_{nm}\mathbf{a}_{mm}\mathbf{v}^{\mathrm{T}}_{nm}\, \forall m, n = \{1,2,..,J\} \equiv \mathcal{J}$, is the target response matrix at the $m^{\mathrm{th}}$ BS due to the transmission from the $n^{\mathrm{th}}$ BS in which $\mathbf{a}_{mm}$ and $\mathbf{v}_{nm}$ are the array response vectors in the directions at the angle-of-arrival $\theta_{mm}$ and the angle of departure $\theta_{nm}$, respectively. Here, $\alpha_{nm}$ represents the complex amplitude of the received signal due to pathloss and the radar cross section of the target and $\mathbf{Z}_{m}^{\mathrm{R}}\in \mathbb{C}^{N_{\mathrm{rx}}\times L}$ is an additive white Gaussian noise (AWGN) vector with the variance of each entry being $\sigma^{2}_{\mathrm{R}}$. Equation \eqref{yr} assumes that all the neighboring BSs have a LoS link to the $m^{\mathrm{th}}$ BS's target. %Please note that the equation can be easily adapted for a case with a subset of BSs having LoS links with the $m^{\mathrm{th}}$ BS's target. The $m^{\mathrm{th}}$ BS transmits a narrowband signal matrix, $\mathbf{X}_{m}\in \mathcal{C}^{N_{\mathrm{tx}}\times L}$, to the users in the cell, with $L > N_{\mathrm{tx}} $ being the length of the radar pulse/ communication frame.
The first term on the right-hand side (RHS) of \eqref{yr} is \color{black}the mono-static \color{black} intra-cell reflection due to the signal vector from the same BS, whereas the second term represents ICR due to the signals from the remaining BSs. \color{black}Note that when $J=1$, \eqref{yr} reduces to the sensing model proposed in \cite{liu2021cramer}\color{black}. \color{black}Furthermore, the transmit antennas at the BS are conventionally down-tilted, resulting in negligible BS-BS interference received through the side lobes of the transmit beampattern \cite{benzaghta2023designing}, and hence in \eqref{yr}, we neglect the direct interference link between the BSs.\color{black}
\subsubsection{Sensing Performance Metric}
\color{black}We assume the target parameter estimation happens locally in the CBF and CoMP modes. Consequently, each BS needs to estimate three parameters of its target: $\boldsymbol{\zeta}_{m}=\{\theta_{mm},\alpha ^{\mathrm{R}}_{mm},\alpha ^{\mathrm{I}}_{mm}\}\equiv \{{\zeta}_{ml}\}$. The sensing process aims to estimate the target parameters using the received echo signal samples $\mathbf{Y}_{m}^{\mathrm{I}}$. We aim to minimize the variance  of the error in the parameter estimation. For an unbiased estimator, the error variance is lower bounded by the CRB given by the inverse of the Fisher information matrix (FIM). Here, the FIM will be a $3\times 3$ matrix  \footnote{\color{black}The extension to multiple targets and multiple parameters can be done using (10) of \cite{li2007range}.} given by \cite{li2007range},
\begin{align}
   \mathbf{F}_{m,*}=2\begin{bmatrix}{F11}_{m,*}^{\mathrm{R}}&{F12}_{m,*}^{\mathrm{R}}&-{F12}_{m,*}^{\mathrm{I}}\\{F12}_{m,*}^{\mathrm{R}}&{F22}_{m,*}^{\mathrm{R}}&-{F22}_{m,*}^{\mathrm{I}}\\-{F12}_{m,*}^{\mathrm{I}}&-{F22}_{m,*}^{\mathrm{I}}&{F22}_{m,*}^{\mathrm{R}}\\\end{bmatrix} \label{fmatrix}
\end{align}
where 
\begin{align}
{Flp}_{m,*} &= \mathrm{tr}\left( \frac{d \boldsymbol{\mu}^{H}_{m,*}}{d {\zeta}_{ml}} \mathbf{C}_{m,*}^{-1} \frac{d \boldsymbol{\mu}_{m,*}}{d {\zeta}_{mp}} \right) \,\, \forall m; l,p \in \{1,2\}; p>=l. \label{fim}
% {F12}_{m,*} &= \mathrm{tr}\left( \frac{d \boldsymbol{\mu}^{H}_{m,*}}{d \boldsymbol{\zeta}_{m1}} \mathbf{C}_{m,*}^{-1} \frac{d \boldsymbol{\mu}_{m,*}}{d {\zeta}_{m2}} \right) \,\, \forall m \in \mathcal{J}.\label{fim1}\\
% {F22}_{m,*} &= \mathrm{tr}\left( \frac{d \boldsymbol{\mu}^{H}_{m,*}}{d \boldsymbol{\zeta}_{m1}} \mathbf{C}_{m,*}^{-1} \frac{d \boldsymbol{\mu}_{m,*}}{d {\zeta}_{m2}} \right) \,\, \forall m \in \mathcal{J}.\label{fim1}
\end{align}
Equation \eqref{fim} is derived from the observation that the received echo signal at the $m^{\mathrm{th}}$ BS is a multi-variate Gaussian random variable with mean $\boldsymbol{\mu}_{m,*}$ and covariance matrix $\mathbf{C}_{m,*}$. The entries of the $\boldsymbol{\mu}_{m,*}$  and $\mathbf{C}_{m,*}$ depend on whether the BSs are operating in the CBF or the CoMP mode, whose corresponding expressions are derived in the following propositions by extending the CRB derivation in \cite{li2007range} to a multi-cell ISAC scenario. \color{black}Additionally, when $\boldsymbol{\zeta}_{m}=\{\theta_{mm}\}$, \eqref{fim} represents the Fisher information value derived in \cite{babu2023multi}\color{black}. %Note that given the generic application of this process to unbiased estimators, we follow the standard in the relevant literature where we do not need to explicitly consider the radar receive processing in \eqref{yr}.
  \subsubsection{\color{black}Coordinated Beamforming for Mono-Static Sensing}
  \color{black}In the case of CBF, ICR acts as an interference term to the serving BS's mono-static sensing. Accordingly, the CRB can be calculated as below\color{black}.
%\subsubsection*{\mathrmbf{Proposition 1}} 
\begin{proposition}
The following equalities hold for the $m^{\mathrm{th}}$ BS in CBF mode with $\dot{\mathbf{a}}_{mm}$ and $\dot{\mathbf{v}}_{mm}$ being the derivatives of ${\mathbf{a}}_{mm}$ and ${\mathbf{v}}_{mm}$ with respect to $\theta_{mm}$:
\begin{align}
&\frac{1}{L\alpha^{2}_{mm}}F11_{m,\mathrm{cbf}}=
\dot{\mathbf{a}}^{H}_{mm} \mathbf{C}_{m,\mathrm{cbf}}^{-1} \dot{\mathbf{a}}_{mm} \cdot  \mathbf{v}^H_{mm}\mathbf{R}^{*}_{\mathbf{X}_m}\mathbf{v}_{mm}\nonumber\\ 
&+ \dot{\mathbf{a}}^{H}_{mm} \mathbf{C}_{m,\mathrm{cbf}}^{-1} {\mathbf{a}}_{mm} \cdot  \mathbf{v}^H_{mm}\mathbf{R}^{*}_{\mathbf{X}_m}\dot{\mathbf{v}}_{mm}\nonumber\\
&+ {\mathbf{a}}^{H}_{mm} \mathbf{C}_{m,\mathrm{cbf}}^{-1} \dot{\mathbf{a}}_{mm} \cdot  \dot{\mathbf{v}}^H_{mm}\mathbf{R}^{*}_{\mathbf{X}_m}\mathbf{v}_{mm}\nonumber\\    
&+{\mathbf{a}}^{H}_{mm} \mathbf{C}_{m,\mathrm{cbf}}^{-1} {\mathbf{a}}_{mm} \cdot  \dot{\mathbf{v}}^H_{mm}\mathbf{R}^{*}_{\mathbf{X}_m}\dot{\mathbf{v}}_{mm}\label{fim_part}\\
&\color{black}\frac{1}{L\alpha_{mm}}F12_{m,\mathrm{cbf}}=
\dot{\mathbf{a}}^{H}_{mm} \mathbf{C}_{m,\mathrm{cbf}}^{-1} {\mathbf{a}}_{mm} \cdot  \mathbf{v}^H_{mm}\mathbf{R}^{*}_{\mathbf{X}_m}\mathbf{v}_{mm}\nonumber\\ 
&\color{black}+ {\mathbf{a}}^{H}_{mm} \mathbf{C}_{m,\mathrm{cbf}}^{-1} {\mathbf{a}}_{mm} \cdot  \dot{\mathbf{v}}^H_{mm}\mathbf{R}^{*}_{\mathbf{X}_m}{\mathbf{v}}_{mm}\label{fim_part1}\\
&\color{black}\frac{1}{L}F22_{m,\mathrm{cbf}}=
{\mathbf{a}}^{H}_{mm} \mathbf{C}_{m,\mathrm{cbf}}^{-1} {\mathbf{a}}_{mm} \cdot  \mathbf{v}^H_{mm}\mathbf{R}^{*}_{\mathbf{X}_m}\mathbf{v}_{mm} \label{fim_part2}
\end{align}
where,
\begin{align}
\mathbf{C}_{m,\mathrm{cbf}} &= L\sum_{n\neq m}^{J} \mathbf{G}_{nm}\mathbf{W}_{n} \mathbf{W}^{H}_{n}\mathbf{G}^{H}_{nm} +  \sigma^{2}_{\mathrm{R}}\mathbf{I}_{N_{r}}, \label{cmcbf}
\end{align}
where $\mathbf{R}_{\mathbf{X}_m}= \frac{1}{L}\mathbf{X}_{m} \mathbf{X}^{\mathrm{H}}_{m} =  \mathbf{W}_{m}  \mathbf{W}^{H}_{m}=\sum_{k=1}^{K}\mathbf{w}_{mk}  \mathbf{w}^{H}_{mk}=\sum_{k=1}^{K}\mathbf{W}_{mk}$.
\end{proposition}
% \subsubsection*{\mathrmbf{Proof}}
\begin{proof}
In CBF, as no data is shared between the BSs, the $m^{\mathrm{th}}$ BS knows only the data symbol matrix $\mathbf{X}_{m}$. Therefore, from \eqref{yr}, $\boldsymbol{\mu}_{m,\mathrm{cbf}}=\mathbf{G}_{mm}\mathbf{X}_{m}$ and $\mathbf{X}_{n}$ act as interference. Using the definitions of $\mathbf{G}_{mm}$, we have
\begin{align}
   & \frac{d \boldsymbol{\mu}_{m,\mathrm{cbf}}}{d {\theta}_{mm}} = \left(\alpha_{mm} \dot{\mathbf{a}}_{mm}{\mathbf{v}}^{T}_{mm}+\alpha_{mm}{\mathbf{a}}_{mm} \dot{\mathbf{v}}^{T}_{mm}\right)\mathbf{X}_{m}, \label{mucbf}\\
   & \color{black}\frac{d \boldsymbol{\mu}_{m,\mathrm{cbf}}}{d {\alpha}^{\mathrm{R}}_{mm}} = {\mathbf{a}}_{mm}{\mathbf{v}}^{T}_{mm}\mathbf{X}_{m}\label{mucbf1}.
\end{align}
 We get \eqref{fim_part}-\eqref{fim_part2} using \eqref{mucbf} and \eqref{mucbf1}  and the cyclic property of the trace operation: $\mathrm{tr}(ABC) =
\mathrm{tr}(BCA)$ in \eqref{fim_part}.
\end{proof}
\subsubsection{\color{black}Coordinated Multipoint Enabling Bi-Static Sensing}
\color{black}In the case of CoMP, ICR becomes useful as an additional source of bi-static sensing, in addition to the serving BS's mono-static sensing. Accordingly, the CRB can be calculated as below\color{black}.
% \subsubsection*{\mathrmbf{Proposition: 2}} 
\begin{proposition}
For the $m^{\mathrm{th}}$ BS configured in CoMP mode, the following equalities hold:
\begin{align}
& F11_{m,\mathrm{comp}}= L\alpha^{2}_{m,m}\dot{\mathbf{a}}^{H}_{mm} \mathbf{C}^{-1}_{m,\mathrm{cmp}} \dot{\mathbf{a}}_{mm} \cdot  \mathbf{v}^{'^H}_{mm}\mathbf{D}_{m}\mathbf{R}^{*}_{\mathbf{X}}\mathbf{D}^{H}_{m}\mathbf{v}^{'}_{mm}\nonumber\\
    &+ L\alpha^{2}_{mm}\dot{\mathbf{a}}^{H}_{mm} \mathbf{C}^{-1}_{m,\mathrm{cmp}} {\mathbf{a}}_{mm} \cdot  {\mathbf{v}}^{'^H}_{mm}\mathbf{D}_{m}\mathbf{R}^{*}_{\mathbf{X}}\mathbf{D}^{H}_{m}\dot{\mathbf{v}}^{'}_{mm}\nonumber\\    
    &+ L\alpha^{2}_{mm}{\mathbf{a}}^{H}_{mm} \mathbf{C}^{-1}_{m,\mathrm{cmp}} \dot{\mathbf{a}}_{mm} \cdot  {{\dot{\mathbf{v}}}}^{'^H}_{mm}\mathbf{D}_{m}\mathbf{R}^{*}_{\mathbf{X}}\mathbf{D}^{H}_{m}\mathbf{v}^{'}_{mm}\nonumber\\   
    &+L\alpha^{2}_{mm}{\mathbf{a}}^{H}_{mm} \mathbf{C}^{-1}_{m,\mathrm{cmp}} {\mathbf{a}}_{mm} \cdot  {{\dot{\mathbf{v}}}}^{'^H}_{mm}\mathbf{D}_{m}\mathbf{R}^{*}_{\mathbf{X}}\mathbf{D}^{H}_{m}{\dot{\mathbf{v}}}^{'}_{mm}\nonumber
    \end{align}
\begin{align}
     &+L\alpha_{mm}\dot{\mathbf{a}}^{H}_{mm} \mathbf{C}^{-1}_{m,\mathrm{cmp}} \dot{\mathbf{a}}_{mm}\cdot \sum_{n\neq m}^{J}\alpha_{nm}{{\mathbf{v}}}^{'^H}_{nm}\mathbf{D}_n\mathbf{R}^{*}_{\mathbf{X}}\mathbf{D}^{H}_{m}\mathbf{v}^{'}_{mm}\nonumber\\      
     &+L\alpha_{mm}{\mathbf{a}}^{H}_{mm} \mathbf{C}^{-1}_{m,\mathrm{cmp}} \dot{\mathbf{a}}_{mm}\cdot \sum_{n\neq m}^{J}\alpha_{nm} {\mathbf{v}}^{'^H}_{nm}\mathbf{D}_n\mathbf{R}^{*}_{\mathbf{X}}\mathbf{D}^{H}_{m}{\dot{\mathbf{v}}}^{'}_{mm}\nonumber\\
     &+L\alpha_{mm}\dot{\mathbf{a}}^{H}_{mm} \mathbf{C}^{-1}_{m,\mathrm{cmp}} \dot{\mathbf{a}}_{mm} \cdot \sum_{n\neq m}^{J}\alpha_{nm}  {{\mathbf{v}}}^{'^H}_{mm}\mathbf{D}_m\mathbf{R}^{*}_{\mathbf{X}}\mathbf{D}^{H}_{n}\mathbf{v}^{'}_{nm}\nonumber\\
     &+L\alpha_{mm}\dot{\mathbf{a}}^{H}_{mm} \mathbf{C}^{-1}_{m,\mathrm{cmp}} {\mathbf{a}}_{mm} \cdot \sum_{n\neq m}^{J}\alpha_{nm} {{\dot{\mathbf{v}}}}^{'^H}_{mm}\mathbf{D}_m\mathbf{R}^{*}_{\mathbf{X}}\mathbf{D}^{H}_{n}\mathbf{v}^{'}_{nm}\nonumber\\
     &+L\alpha_{mm}\dot{\mathbf{a}}^{H}_{mm} \mathbf{C}^{-1}_{m,\mathrm{cmp}}\dot{\mathbf{a}}_{mm} \cdot \sum_{n\neq m}^{J}\sum_{n\neq m}^{J}\alpha_{nm} {{\mathbf{v}}}^{'^ H}_{nm}\mathbf{D}_n\mathbf{R}^{*}_{\mathbf{X}}\mathbf{D}^{H}_{n}\mathbf{v}^{'}_{nm}, \label{compfim}\\  
     &\color{black} F12_{m,\mathrm{comp}}= L\alpha_{mm}\dot{\mathbf{a}}^{H}_{mm} \mathbf{C}^{-1}_{m,\mathrm{cmp}} {\mathbf{a}}_{mm} \cdot  \mathbf{v}^{'^H}_{mm}\mathbf{D}_{m}\mathbf{R}^{*}_{\mathbf{X}}\mathbf{D}^{H}_{m}\mathbf{v}^{'}_{mm}\nonumber\\
     &\color{black}+ L\alpha_{mm}{\mathbf{a}}^{H}_{mm} \mathbf{C}^{-1}_{m,\mathrm{cmp}} {\mathbf{a}}_{mm} \cdot  \dot{\mathbf{v}}^{'^H}_{mm}\mathbf{D}_{m}\mathbf{R}^{*}_{\mathbf{X}}\mathbf{D}^{H}_{m}{\mathbf{v}}^{'}_{mm}\nonumber\\ 
      &\color{black}+ L\dot{\mathbf{a}}^{H}_{mm} \mathbf{C}^{-1}_{m,\mathrm{cmp}}{\mathbf{a}}_{mm}\cdot \sum_{n\neq m}^{J}\alpha_{nm} {\mathbf{v}}^{'^H}_{nm}\mathbf{D}_n\mathbf{R}^{*}_{\mathbf{X}}\mathbf{D}^{H}_{m}\mathbf{v}^{'}_{mm} \label{compfim1}\\
     &\color{black} F22_{m,\mathrm{comp}} = L{\mathbf{a}}^{H}_{mm} \mathbf{C}^{-1}_{m,\mathrm{cmp}} {\mathbf{a}}_{mm} \cdot  \mathbf{v}^{'^H}_{mm}\mathbf{D}_{m}\mathbf{R}^{*}_{\mathbf{X}}\mathbf{D}^{H}_{m}\mathbf{v}^{'}_{mm}\label{compfim2}
\end{align}
where $\mathbf{D}_{m}=\mathrm{diag}( \mathbf{0}_{N_{\mathrm{tx}}},.., \mathbf{I}_{N_\mathrm{tx}},.., \mathbf{0}_{N_{\mathrm{tx}}})\, \in \mathcal{C}^{N \times N}$; $\mathbf{v}^{'}_{mm}=\{\Vec{0}_{N_{\mathrm{tx}}},..,\mathbf{v}_{mm},..,\Vec{0}_{N_{\mathrm{tx}}}\}\, \in \mathcal{C}^{N \times 1}$.
\end{proposition}
% \subsubsection*{\mathrmbf{Proof}}
\begin{proof}
Since each BS knows the transmitted symbol matrix of the other BSs, the ICR of \eqref{yr} aids the target's angle estimation in a bi-static manner. Consequently, the received echo at the $m^{\mathrm{th}}$ BS is a multi-variate Gaussian random variable with mean $\boldsymbol{\mu}_{m,\mathrm{cmp}} = \mathbf{G}^{'}_{mm}\mathbf{D}_{m}\mathbf{X}+\sum_{n\neq m}^{J} \mathbf{G}^{'}_{nm}\mathbf{D}_{n}\mathbf{X}$ where $\mathbf{G}^{'}_{nm}=\alpha_{nm}\mathbf{a}_{mm}\mathbf{v}^{{'}^{\mathrm{T}}}_{nm}$
and the covariance matrix $\mathbf{C}_{m,\mathrm{cmp}} =  \sigma^{2}_{\mathrm{R}}\mathbf{I}_{N_{\mathrm{rx}}}$. Hence,
\begin{align}
   & \frac{d \boldsymbol{\mu}_{m,\mathrm{cmp}}}{d \mathbf{\theta}_{mm}} = \alpha_{mm} \left[\dot{\mathbf{a}}_{mm}{\mathbf{v}}^{'^T}_{mm}+{\mathbf{a}}_{mm} \dot{\mathbf{v}}^{'^T}_{mm}\right]\mathbf{D}_{m}\mathbf{X}\nonumber\\
   & +\sum_{n\neq m}^{J}\alpha_{nm}\dot{\mathbf{a}}_{mm}{\mathbf{v}}^{'^T}_{nm}\mathbf{D}_{n}\mathbf{X}, \label{mucomp}\\
   &\color{black}\frac{d \boldsymbol{\mu}_{m,\mathrm{cmp}}}{d \mathbf{\alpha}^{\mathrm{R}}_{mm}} =  \left[{\mathbf{a}}_{mm}{\mathbf{v}}^{'^T}_{mm}\right]\mathbf{D}_{m}\mathbf{X}.\label{mucomp1}
\end{align}
 Equations \eqref{compfim}- \eqref{compfim2} are obtained using \eqref{mucomp} and \eqref{mucomp1} and the cyclic property of the trace operation in \eqref{fim_part}.
 \end{proof}
 
\color{black}The corresponding Fisher Information matrices, $\mathbf{F}_{m,\mathrm{cbf}}$ and $\mathbf{F}_{m,\mathrm{comp}}$,  are obtained using \eqref{fim_part}-\eqref{fim_part2} and \eqref{compfim}- \eqref{compfim2} in \eqref{fmatrix}\color{black}. The following problem formulations assume knowledge of target locations and ICR directions: $\theta_{nm}$ $\forall m,n \in {1,2,.., J}$. This is quite a typical assumption in the radar literature and can be interpreted as optimizing the precoders towards a specific direction of interest \cite{li2007range}.
\section{Robust Block Level Precoding}\label{BLP}
The BLP constrains the co-channel interference experienced by a user so that the received symbol is within a certain distance from the nominal constellation symbol. \color{black}In this section, we explain the BLP design framework for CBF and CoMP BS coordination schemes.\color{black}
\subsection{BLP: Coordinated Beamforming} \label{seccbf}
We take a worst-case approach for the transmit precoding design to guarantee the resulting solution is robust to all possible channel uncertainties within $\mathcal{E}_{mk}$. Hence, in the CBF mode, for the $m^{\mathrm{th}}$ BS, we aim to solve the following optimization problem:
\begin{subequations}
\begin{align}
\color{black}\mathrm{(P1):} &\underset{\{\mathbf{w}_{mk}\}}{\mathrm{maximize}}\,\,\,\,  \frac{\rho}{\mathrm{NF}^{\mathrm{blp}}_{\mathrm{R},{\mathrm{cbf}}}} \,\, {-t^{1}_{m,\mathrm{cbf}}} +\frac{(1-\rho)}{\mathrm{NF}^{\mathrm{blp}}_{\mathrm{C},{\mathrm{cbf}}}} \gamma,\nonumber\\
\mathrm{s.t.}&\color{black}\begin{bmatrix}
\mathbf{F}_{m,\mathrm{cbf}}& \mathbf{I}_{:l}\\ \mathbf{I}^{T}_{:l}  
 & t^{l}_{m,\mathrm{cbf}} \\
\end{bmatrix} \succeq 0 \quad \forall l\in\{1,2,3\} \quad \forall m\label{cbfschur}\\
&\mathbf{v}^{\mathrm{T}}_{mn} \sum_{k}^{K} \mathbf{w}_{mk}\mathbf{w}^{\mathrm{H}}_{mk} \mathbf{v}^{*}_{mn} \leq P_{\mathrm{leak}}, \quad \forall n \label{leakage}\\
&\color{black}\underset{\mathbf{e}_{mk}}{\mathrm{min}}\left(\gamma_{mk}\right) \geq \gamma \quad \forall U_{mk},\label{p1.sinr}\\
&\sum_{k=1}^{K}\mathrm{tr}\left ( \mathbf{w}_{mk} \mathbf{w}^{H}_{mk} \right ) \leq P_{\mathrm{t}}, \quad \forall m\label{p1.pt}
\end{align}
\end{subequations}
where 
\begin{align}
   \color{black} \gamma_{mk}=\frac{|\mathbf{\tilde{h}}^{T}_{m,mk}\mathbf{w}_{mk}|^{2}}{\underbrace{\sum_{l\neq k}^{K}|\mathbf{\tilde{h}}^{T}_{m,mk}\mathbf{w}_{ml}|^{2}}_{\mathrm{InCI}_{mk}}+\sum_{n\neq m}^{J}\underbrace{\sum_{l=1}^{K}|\mathbf{\tilde{h}}^{T}_{n,mk}\mathbf{w}_{nl}|^{2}}_{\mathrm{ICI}_{n,mk}}+\sigma^{2}_{\mathrm{C}}},
\end{align}
\color{black}is the received SINR at $U_{mk}$ obtained using \eqref{yc}\color{black}. The objective function of (P1) is the weighted sum of two components: minimizing the CRB in estimating the target's angle and maximizing the minimum SINR value among the users of the $m^{\mathrm{th}}$ cell. The weighting factor, $\rho \in \left[0,1\right]$, determines the balance between the communication and sensing performance matrices. The normalization factors $\mathrm{NF}^{\mathrm{blp}}_{\mathrm{R},{\mathrm{cbf}}}$ and $\mathrm{NF}^{\mathrm{blp}}_{\mathrm{C},{\mathrm{cbf}}}$ are determined by solving (P1) for $\rho=1$ and $\rho=0$, respectively. Using \eqref{cbfschur}, we limit the $l^{\mathrm{th}}$ diagonal of $\mathbf{F}^{-1}_{m,\mathrm{cbf}}$ to be less than or equal to $t^{l}_{m,\mathrm{cbf}}$ using the Schur complement: $t^{l}_{m,\mathrm{cbf}}-\mathbf{I}^{T}_{:l}\mathbf{F}^{-1}_{m,\mathrm{cbf}}\mathbf{I}_{:l} >= 0$, where $\mathbf{I}_{:l}$ is the $l^{\mathrm{th}}$ column of an identity matrix\color{black}. Equation \eqref{leakage} restricts the power leaked towards the  $n^{\mathrm{th}}$ BS's target to be less than or equal to $P_{\mathrm{leak}}$. Ideally, we want $P_{\mathrm{leak}}$ to be zero in the CBF scenario; however, it is not feasible if a user is located at $\theta_{mn}$ direction. Moreover, \eqref{p1.sinr} is the SINR constraint derived from \eqref{yc}, whereas \eqref{p1.pt} is the total power constraint with $P_{\mathrm{t}}$ being the total available power at the BS. Problem (P1) is difficult to solve because of (a) the infinite possibilities of $\mathbf{e}_{mk}$ in \eqref{p1.sinr}, (b) the non-convex multiplication between $\gamma$ and the interference terms in \eqref{p1.sinr} and (c) the non-convex form of $\mathbf{C}_{m,\mathrm{cbf}}$ in \eqref{cbfschur}.

\color{black}We tackle the infinite possibilities of $\mathbf{e}_{mk}$ by representing it as the ratio of the minimum of the numerator to the maximum of the denominator:
% \begin{align}       
% &\frac{\underset{\mathbf{e}_{mk}}{\mathrm{min}}\, (|\mathbf{\tilde{h}}^{T}_{m,mk}\mathbf{w}_{mk}|^{2})}{\underset{\mathbf{e}_{mk}}{\mathrm{max}} (\sum_{l\neq k}^{K}|\mathbf{\tilde{h}}^{T}_{m,mk}\mathbf{w}_{ml}|^{2}+\sum_{n\neq m}^{J}\sum_{l=1}^{K}|\mathbf{\tilde{h}}^{T}_{n,mk}\mathbf{w}_{nl}|^{2}+\sigma^{2}_{\mathrm{C}})}\nonumber\\ &\geq \gamma\,\, \forall U_{mk}. \label{min}
% \end{align}
\begin{align}       
&\frac{\underset{\mathbf{e}_{mk}}{\mathrm{min}}\, (|\mathbf{\tilde{h}}^{T}_{m,mk}\mathbf{w}_{mk}|^{2})}{\underset{\mathbf{e}_{mk}}{\mathrm{max}} (\mathrm{InCI}_{mk}+\sum_{n\neq m}^{J}\mathrm{ICI}_{n,mk}+\sigma^{2}_{\mathrm{C}})} \geq \gamma. \label{min}
\end{align}
The minimum value of the numerator and the maximum value of the denominator of \eqref{min} can be determined using Proposition 3.\color{black}
% \subsubsection*{\mathrmbf{Proposition 3}}
\begin{proposition}
For a given CSI uncertainty set $\mathcal{E}_{mk}=\{\mathbf{e}_{mk}:\|\mathbf{e}_{mk}\|^{2}\leq \delta^{2}\}$, the following equalities hold  for any $U_{mk}$.
\begin{align}
    &\underset{\mathbf{e}_{mk}}{\mathrm{min}} |\mathbf{\tilde{h}}^{T}_{m,mk}\mathbf{w}_{mk}|^{2}=\mathrm{tr}\left (  \mathbf{{Q}}_{m,mk}  \mathbf{W}_{mk} \right ),
    \end{align}    
    \begin{align}    
    \underset{\mathbf{e}_{mk}}{\mathrm{max}} \left(|\mathbf{\tilde{h}}^{T}_{m,mk}\mathbf{w}_{mk}|^{2}\right)
    &= \mathrm{tr}\left (  \mathbf{{Q}}_{m,mk}  \mathbf{W}_{mk} \right )+ \delta^{2} \mathrm{tr}\left (\mathbf{W}_{mk}\right) \nonumber\\&+ \delta \|\mathbf{{h}}^{T}_{m,mk} \mathbf{W}_{mk}\|
 + \delta \|\mathbf{W}_{mk} \mathbf{{h}}^{*}_{m,mk}\|,
\end{align} 
where $\mathbf{Q}_{m,mk}= \mathbf{h}^{*}_{m,mk} \mathbf{h}^{T}_{m,mk}$ and $\mathbf{W}_{mk}= \mathbf{w}_{mk}\mathbf{w}^{H}_{mk}$.
\end{proposition}
\begin{proof}    
 Let,
\begin{align}
    &|\mathbf{\tilde{h}}^{T}_{m,mk}\mathbf{w}_{mk}|^{2} = \mathrm{tr}\left (  \mathbf{\tilde{Q}}_{m,mk}  \mathbf{W}_{mk} \right )\\ 
     &=  \mathrm{tr}\left (  \mathbf{\tilde{h}}^{*}_{m,mk} \mathbf{\tilde{h}}^{T}_{m,mk} \mathbf{W}_{mk} \right )\\
     &= \mathrm{tr}\left (   \mathbf{\tilde{h}}^{T}_{m,mk} \mathbf{W}_{mk} \mathbf{\tilde{h}}^{*}_{m,mk} \right )\\
     &= \mathrm{tr}\left (   \left(\mathbf{{h}}^{T}_{m,mk}+\mathbf{{e}}^{T}_{mk} \right) \mathbf{W}_{mk} \left(\mathbf{{h}}^{*}_{m,mk}+\mathbf{{e}}^{*}_{mk}\right) \right ) \\
     &= \mathrm{tr}\left ( \mathbf{{h}}^{T}_{m,mk} \mathbf{W}_{mk} \mathbf{{h}}^{*}_{m,mk} \right)+\mathrm{tr}\left ( \mathbf{{h}}^{T}_{m,mk} \mathbf{W}_{mk} \mathbf{{e}}^{*}_{mk}\right)\nonumber\\
     &+\mathrm{tr}\left (\mathbf{{e}}^{T}_{mk} \mathbf{W}_{mk} \mathbf{{h}}^{*}_{m,mk}\right)+\mathrm{tr}\left (\mathbf{{e}}^{T}_{mk} \mathbf{W}_{mk} \mathbf{{e}}^{*}_{mk}\right). 
\end{align}
Now,
\begin{align}
&\underset{\|\mathbf{e}_{mk}\|^{2}\leq \delta^{2}}{\mathrm{max}} \mathrm{tr}\left (  \mathbf{\tilde{Q}}_{m,mk}  \mathbf{W}_{mk} \right ) = \mathrm{tr}\left (  \mathbf{{Q}}_{m,mk}  \mathbf{W}_{mk} \right ) \nonumber\\
& + \|\mathbf{{h}}^{T}_{m,mk} \mathbf{W}_{mk} 
\mathbf{{e}}^{*}_{mk}\|+ \|\mathbf{{e}}^{T}_{mk} \mathbf{W}_{mk} \mathbf{{h}}^{*}_{m,mk}\|\nonumber\\
&+ \mathrm{tr}\left (\mathbf{{e}}^{*}_{mk}\mathbf{{e}}^{T}_{mk} \mathbf{W}_{mk}\right) \label{max1}
\end{align}
\begin{align}
& \leq \mathrm{tr}\left (  \mathbf{{Q}}_{m,mk}  \mathbf{W}_{mk} \right ) + \delta \|\mathbf{{h}}^{T}_{m,mk} \mathbf{W}_{mk} 
\|\nonumber\\
& + \delta \|\mathbf{W}_{mk} \mathbf{{h}}^{*}_{m,mk}\|+ \mathrm{tr}\left (\mathbf{{e}}^{*}_{mk}\mathbf{{e}}^{T}_{mk}\right) \mathrm{tr}\left (\mathbf{W}_{mk}\right) \nonumber\\
& \leq \mathrm{tr}\left (  \mathbf{{Q}}_{m,mk}  \mathbf{W}_{mk} \right ) + \delta \|\mathbf{{h}}^{T}_{m,mk} \mathbf{W}_{mk} 
\|\nonumber\\
& + \delta \|\mathbf{W}_{mk} \mathbf{{h}}^{*}_{m,mk}\|+ \delta^{2} \mathrm{tr}\left (\mathbf{W}_{mk}\right). \label{max2}
\end{align}
Here, \eqref{max2} is derived from \eqref{max1} using Cauchy-Schwarz inequality. The minimum value of the numerator is obtained by substituting $\delta =0$ in \eqref{max2}.
\end{proof}
\color{black}Consequently, 
\begin{align}
  \underset{\mathbf{e}_{mk}}  {\mathrm{max}} (\mathrm{InCI}_{mk})&= \sum_{l\neq k}^{K} \mathrm{tr}\left (  \mathbf{{Q}}_{m,mk}  \mathbf{W}_{ml} \right ) + \delta \|\mathbf{{h}}^{T}_{m,mk} \mathbf{W}_{ml} \|  \nonumber\\&+ \delta \|\mathbf{W}_{ml} \mathbf{{h}}^{*}_{m,mk}\| + \delta^{2} \mathrm{tr}\left (\mathbf{W}_{ml}\right)\\
  \underset{\mathbf{e}_{mk}}  {\mathrm{max}} (\mathrm{ICI}_{n,mk})&= \sum_{l=1}^{K} \mathrm{tr}\left (  \mathbf{{Q}}_{n,mk}  \mathbf{W}_{nl} \right ) + \delta \|\mathbf{{h}}^{T}_{n,mk} \mathbf{W}_{nl} \|  \nonumber\\&+ \delta \|\mathbf{W}_{nl} \mathbf{{h}}^{*}_{n,mk}\|
+ \delta^{2} \mathrm{tr}\left (\mathbf{W}_{nl}\right) 
\end{align}
The non-convex multiplication between $\gamma$ and the interference terms in \eqref{p1.sinr} is addressed by representing \eqref{p1.sinr} through a set of intra-cell and inter-cell leakage constraints \eqref{p1.1.c2}-\eqref{p1.1.c1}:
\begin{align}
     &\underset{\mathbf{e}_{mk}}  {\mathrm{max}} (\mathrm{InCI}_{mk}) \leq I^{\mathrm{intra}}_{m,\mathrm{cbf}},\,\, \forall U_{mk}\label{p1.1.c2} \\
     & \underset{\mathbf{e}_{nk}}  {\mathrm{max}} (\mathrm{ICI}_{m,nk}) \leq \frac{I^{\mathrm{inter}}_{\mathrm{cbf}}}{J-1}\,\, \forall U_{nk}\, n\neq m\label{p1.1.c2.2}\\
     & \mathrm{tr}\left (  \mathbf{{Q}}_{m,mk}  \mathbf{W}_{mk} \right ) - \gamma (I^{\mathrm{intra}}_{m,\mathrm{cbf}}+I^{\mathrm{inter}}_{\mathrm{cbf}}) \geq \gamma \sigma^{2}_{\mathrm{C}} \,\, \forall U_{mk}.    \label{p1.1.c1}
\end{align}
% where,
% \begin{align}
%     \underset{\mathbf{e}_{nk}}  {\mathrm{max}} (\mathrm{ICL}_{m,nk})= \sum_{l=1}^{K} \mathrm{tr}\left (  \mathbf{{Q}}_{m,nk}  \mathbf{W}_{ml} \right ) + \delta \|\mathbf{{h}}^{T}_{m,nk} \mathbf{W}_{ml} \|  \nonumber\\+ \delta \|\mathbf{W}_{ml} \mathbf{{h}}^{*}_{m,nk}\|
% + \delta^{2} \mathrm{tr}\left (\mathbf{W}_{ml}\right)\label{icnl}.
% \end{align}
Equation \eqref{p1.1.c2} limits the intracell interference to any user served by $m^{\mathrm{th}}$ BS to be less than $I^{\mathrm{intra}}_{m,\mathrm{cbf}}$. The left-hand-side (LHS) of \eqref{p1.1.c2.2} gives the maximum ICI from $m^{\mathrm{th}}$ BS to $U_{nk}$. The constraint \eqref{p1.1.c2.2} satisfied by all the BSs restricts the ICI to any user in the system to $I^{\mathrm{inter}}_{\mathrm{cbf}}$. Note that for given $I^{\mathrm{intra}}_{m,\mathrm{cbf}}$ and $I^{\mathrm{inter}}_{\mathrm{cbf}}$ values, the constraints \eqref{p1.1.c2}-\eqref{p1.1.c1} become convex. Lastly, for a given $\{\mathbf{W}_{nk}\}$, $\mathbf{C}_{m,\mathrm{cbf}}$ can be estimated using \eqref{cmcbf}, \color{black}which makes the entries of $\mathbf{F}_{m,\mathrm{cbf}}$ of \eqref{cbfschur}, an affine function function of $\mathbf{R}_{\mathbf{X}_{m}}=\sum_{k}^{K}\mathbf{W}_{mk}$.
Hence, we solve (P1) by solving the following two optimization problems alternatively:
\begin{subequations}
\begin{align}
\color{black}\mathrm{(P1.A):}& \underset{\{\mathbf{W}_{mk}\}}{\mathrm{maximize}}\,\,\,\,  \frac{\rho}{\mathrm{NF}^{\mathrm{blp}}_{\mathrm{R},{\mathrm{cbf}}}} \,\, {-t^{1}_{m,\mathrm{cbf}}} +\frac{(1-\rho)}{\mathrm{NF}^{\mathrm{blp}}_{\mathrm{C},{\mathrm{cbf}}}} \gamma,\nonumber\\
\mathrm{s.t.}& \mathbf{v}^{\mathrm{T}}_{mn} \sum_{k}^{K} \mathbf{W}_{mk} \mathbf{v}^{*}_{mn} \leq P_{\mathrm{leak}}, \quad \forall n,\label{p1.1.c1.1}\\
& \sum_{k=1}^{K}\mathrm{tr}\left ( \mathbf{W}_{mk} \right ) \leq P_{\mathrm{t}},\label{p1.1.c3}\\
&\mathrm{rank}(\mathbf{W}_{mk}) = 1; \, \mathbf{W}_{mk} \succeq \mathbf{0}, \label{p1.1.c4}\\
&\eqref{cbfschur}, \eqref{p1.1.c2}-\eqref{p1.1.c1}.
% & \eqref{p1.1.c1.1}-\eqref{p1.1.c4};\label{p.1.1.a.2}
\end{align}
\end{subequations}
\begin{subequations}
\begin{align}
\mathrm{(P1.B):} &\,\,\underset{\{\mathbf{W}_{mk}\}}{\mathrm{minimize}}\quad  \frac{P_{\mathrm{leak}}}{P^{\mathrm{max}}_{\mathrm{leak}}}+\frac{I^{{\mathrm{intra}}}_{m,\mathrm{cbf}}}{I^{{\mathrm{intra}}}_{\mathrm{max}}}+\frac{I^{{\mathrm{inter}}}_{m,\mathrm{cbf}}}{(J-1)I^{{\mathrm{inter}}}_{\mathrm{max}}},\nonumber\\
 \mathrm{s.t.} &\color{black}\begin{bmatrix}
\mathbf{F}_{m,\mathrm{cbf}}& \mathbf{I}_{:l}\\ \mathbf{I}^{T}_{:l}  
 &\color{black} t^{*l}_{m,\mathrm{cbf}} \\
\end{bmatrix} \succeq 0 \quad \forall l\in\{1,2,3\} \quad \forall m, \label{p1.b.1}\\
% & \mathbf{v}^{\mathrm{T}}_{mn} \sum_{k}^{K} \mathbf{W}_{mk} \mathbf{v}^{*}_{mn} \leq P_{\mathrm{l}}, \quad \forall n \\
&\mathrm{tr}\left (  \mathbf{Q}_{m,mk}  \mathbf{W}_{mk} \right )-\gamma^{*}\left(I^{{\mathrm{intra}}}_{m,\mathrm{cbf}} + I^{{\mathrm{inter}}}_{\mathrm{cbf}} \right) \geq \gamma^{*} \sigma^{2}_{R},\label{p1.b.2}\\
&\underset{\mathbf{e}_{nk}}  {\mathrm{max}} (\mathrm{ICI}_{m,nk})  \leq \frac{I^{\mathrm{inter}}_{m,\mathrm{cbf}}}{J-1};\,\, I^{\mathrm{inter}}_{m,\mathrm{cbf}} \leq I^{\mathrm{inter}}_{\mathrm{cbf}} \\
&  \eqref{p1.1.c1.1}, \eqref{p1.1.c3}, \eqref{p1.1.c2}, \eqref{p1.1.c4},
\end{align}
\end{subequations}
\color{black}The first problem, (P1.A), minimizes the CRB and maximizes the minimum communication SINR for a given $\{\mathbf{W}_{nk}\}$ $\forall n \neq m \in \mathcal{J}$, $P_{\mathrm{leak}}$, $I^{\mathrm{intra}}_{m,\mathrm{cbf}}$, and $I^{\mathrm{inter}}_{\mathrm{cbf}}$ values. Equations \eqref{p1.1.c1.1} and \eqref{p1.1.c3} are the equivalent representation of \eqref{leakage} and \eqref{p1.pt}, respectively. Using the solution of (P1.1.A): $\gamma^{*}$, and $\mathrm{t}^{*1}_{m,\mathrm{cbf}}$, \color{black}(P1.B) minimizes the leakage power towards the neighboring BSs' targets and users while guaranteeing given sensing and communication performance through \eqref{p1.b.1} and \eqref{p1.b.2}. Here, $P^{\mathrm{max}}_{\mathrm{leak}}$, $I^{{\mathrm{intra}}}_{\mathrm{max}}$ and $I^{{\mathrm{intra}}}_{\mathrm{max}}$ represent the maximum tolerable leakage power, intra-cell and inter-cell interference values, respectively. \color{black}Omitting the rank constraint, (P1.A) and (P1.B) are convex optimization problems solved using MATLAB's CVX solver \cite{grant2014cvx}. The overall procedure is given in Algorithm \ref{algo1}. We can obtain $\mathbf{w}_{mk}$ from $\mathbf{W}_{mk}$ using Eigenvalue decomposition technique if the rank of $\mathbf{W}_{mk}\,>\,1$  \cite{ozkaptan2023software}. %The initial values of $\{\mathbf{W}_{mk}\}$ can be obtained by setting high values for $I^{\mathrm{intra}}_{m,\mathrm{cbf}}$, and $I^{\mathrm{inter}}_{nm,\mathrm{cbf}}$ and solving the following feasibility problem:
% \begin{align}
% &\mathrm{(P1.1.B):} \underset{\{\mathbf{W}_{mk}\}}{\mathrm{minimize}} 0
% & \gamma 
%    & \eqref{p1.1.c1.1}-\eqref{p1.1.c4}
% \end{align}
\subsection{BLP: Coordinated Multi-point}\label{seccomp}
The corresponding optimization problem in the CoMP scenario can be formulated as,
\begin{subequations}
\begin{align}
\mathrm{(P2):}&  \underset{\{\mathbf{W}_{k}\}}{\mathrm{maximize}}\,\,\,\,   \,\, \frac{-f\rho}{\mathrm{NF}^{\mathrm{blp}}_{\mathrm{R},{\mathrm{comp}}}} + \frac{(1-\rho)\gamma}{\mathrm{NF}^{\mathrm{blp}}_{\mathrm{C},{\mathrm{comp}}}},\nonumber\\
\mathrm{s.t.}&\color{black}\begin{bmatrix}
\mathbf{F}_{m,\mathrm{comp}}& \mathbf{I}_{:l}\\ \mathbf{I}^{T}_{:l}  
 & t^{l}_{m,\mathrm{comp}} \label{compschur1}\\
\end{bmatrix} \succeq 0 \quad \forall l\in\{1,2,3\} \quad \forall m\\
&\color{black}t^{1}_{m,\mathrm{comp}}\leq f, \quad \forall m,\label{compschur2}
\end{align}
 \begin{align}
& \underset{\mathbf{e}_{k}}{\mathrm{min}}\underbrace{\left(\frac{|\mathbf{\tilde{h}}^{T}_{k}\mathbf{w}_{k}|^{2}}{\sum_{l\neq k}^{JK}|\mathbf{\tilde{h}}^{T}_{k}\mathbf{w}_{l}|^{2}+\sigma^{2}_{\mathrm{C}}}\right)}_{\gamma_{k}} \geq \gamma, \,\, \forall k, \label{p2.c2}\\
% \end{align}
% \begin{align}
& \sum_{k=1}^{K}\mathrm{tr}\left ( \mathbf{D}_{m}\mathbf{w}_{k}\mathbf{w}^{H}_{k}\mathbf{D}^{H}_{m} \right ) \leq  P_{\mathrm{t}}, \forall m .\label{p2.c3}  
\end{align}
\end{subequations}
The objective function of (P2) is the weighted combination of minimizing the maximum CRB ($f$) and maximizing the minimum SINR ($\gamma$) values. \color{black}Using \eqref{compschur1}, we limit the $l^{\mathrm{th}}$ diagonal of $\mathbf{F}^{-1}_{m,\mathrm{comp}}$ to be less than or equal to $t^{l}_{m,\mathrm{comp}}$ using the Schur complement. In \eqref{compschur2}, $f$ is defined as the maximum CRB for estimating $\theta_{mm}$. Equation \eqref{p2.c2} is the minimum SINR constraint written using \eqref{yccomp} whereas \eqref{p2.c3} is the per-BS power constraint. The constants $\mathrm{NF}^{\mathrm{blp}}_{\mathrm{R},{\mathrm{comp}}}$ and $\mathrm{NF}^{\mathrm{blp}}_{\mathrm{C},{\mathrm{comp}}}$ are the corresponding maximum values of $f$ and $\gamma$ obtained by setting $\rho=1$ and $\rho=0$, respectively. Here, $\mathbf{w_k}=\left[\mathbf{w}_{1k};\mathbf{w}_{2k;...;\mathbf{w}_{Jk}}\right]\mathcal{C}^{N\times 1}$ represents the precoding vectors for user $k$ from all the BSs stacked vertically. \color{black}Problem (P2) is non-convex due to the infinite possibilities of $\mathbf{e}_{k}$ and the non-convex SINR constraint \eqref{p2.c2}. \color{black}Similar to the CBF case, we tackle the infinite possibilities of $\mathbf{e}_{k}$ by adapting Proposition 3 to the CoMP case. Using Proposition 3, we have
\begin{align}
    &\underset{\mathbf{e}_{k}}{\mathrm{min}} |\mathbf{\tilde{h}}^{T}_{k}\mathbf{w}_{k}|^{2}=\mathrm{tr}\left (  \mathbf{{Q}}_{k}  \mathbf{W}_{k} \right ),
    \end{align}    
    \begin{align}    
    \underset{\mathbf{e}_{k}}{\mathrm{max}} \left(|\mathbf{\tilde{h}}^{T}_{k}\mathbf{w}_{l}|^{2}\right)
    &= \mathrm{tr}\left (  \mathbf{{Q}}_{k}  \mathbf{W}_{l} \right )+ J\delta^{2} \mathrm{tr}\left (\mathbf{W}_{l}\right) \nonumber\\&+ \sqrt{J}\delta \|\mathbf{{h}}^{T}_{k} \mathbf{W}_{l}\|
 + \sqrt{J}\delta \|\mathbf{W}_{l} \mathbf{{h}}^{*}_{k}\|,
\end{align}
where $\mathbf{Q}_{k}=\mathbf{h}_{k}\mathbf{h}^{H}_{k}\,\in \mathcal{C}^{N \times N}$, $\mathbf{W}_{k}= \mathbf{w}_{k}\mathbf{w}^{H}_{k}\,\in \mathcal{C}^{N \times N}$, and $\mathbf{R}_{\mathbf{X}}= \sum_{k=1}^{2K}\mathbf{W}_{k}$.
Then, \eqref{p2.c2} is reformulated into the following pair of constraints.
\begin{align}
 & \mathrm{tr}\left (  \mathbf{{Q}}_{k}  \mathbf{W}_{k} \right ) - \gamma I_{\mathrm{comp}} \geq \gamma \sigma^{2}_{\mathrm{C}}\label{p2.sinr},\\ 
&\sum_{l\neq k}^{JK}\underset{\mathbf{e}_{k}}{\mathrm{max}} \left(|\mathbf{\tilde{h}}^{T}_{k}\mathbf{w}_{l}|^{2}\right) \leq I_{\mathrm{comp}} \label{compleakage}.
\end{align}
Equation \eqref{compleakage} limits the total interference experienced by any user to be less than $I_{\mathrm{comp}}$. For a given $I_{\mathrm{comp}}$ value, \eqref{p2.sinr} and \eqref{compleakage} become convex constraints. Hence, we solve (P2) by solving the (P2.A) and (P2.B) alternately. \color{black}
% \begin{subequations}
% \begin{align}
% &\color{black}\mathrm{(P2.1):}   \underset{\{\mathbf{W}_{k}\}}{\mathrm{maximize}}\,\,\,\,   \,\, \frac{-f\rho}{\mathrm{NF}^{\mathrm{blp}}_{\mathrm{R},{\mathrm{comp}}}} + \frac{(1-\rho)\gamma}{\mathrm{NF}^{\mathrm{blp}}_{\mathrm{C},{\mathrm{comp}}}},\nonumber\\
% &\color{black} \eqref{compschur1}, \eqref{compschur2}\label{p2.1.c1}\\
% & \mathrm{tr}\left (  \mathbf{{Q}}_{k}  \mathbf{W}_{k} \right ) - \gamma I_{\mathrm{comp}} \geq \gamma \sigma^{2}_{\mathrm{C}}\label{p2.1.c2},
%  \end{align}
% \begin{align}
% &\sum_{l\neq k}^{K} \mathrm{tr}\left (  \mathbf{{Q}}_{k}  \mathbf{W}_{k} \right ) + \sqrt{J}\delta \|\mathbf{{h}}^{T}_{k} \mathbf{W}_{k} \|  + \sqrt{J}\delta \|\mathbf{W}_{k} \mathbf{{h}}^{*}_{k}\|\nonumber\\
% &+ {J}\delta^{2} \mathrm{tr}\left (\mathbf{W}_{k}\right) \leq I_{\mathrm{comp}},\\
% & \sum_{k=1}^{K}\mathrm{tr}\left ( \mathbf{D}_{m}\mathbf{W}_{k}\mathbf{D}^{H}_{m} \right ) \leq  P_{\mathrm{t}}, \label{p2.1.c3} \\
% &\mathbf{W}_{k} \succeq \mathbf{0} \, \forall k, \, \mathrm{rank}(\mathbf{W}_{k}) = 1,  \label{p2.1.c4}
% \end{align}
% \end{subequations}
% Also, $\mathbf{R}_{\mathbf{X}}= \sum_{k=1}^{2K}\mathbf{W}_{k}$ and  For a given $I_{\mathrm{comp}}$, (P2.1) a convex optimization problem and is solved by solving the following two optimization problems alternatively:
\begin{subequations}
\begin{align}
\color{black}\mathrm{(P2.A):}&  \underset{\{\mathbf{W}_{k}\}}{\mathrm{maximize}}\,\,\,\,   \,\, \frac{-f\rho}{\mathrm{NF}^{\mathrm{blp}}_{\mathrm{R},{\mathrm{comp}}}} + \frac{(1-\rho)\gamma}{\mathrm{NF}^{\mathrm{blp}}_{\mathrm{C},{\mathrm{comp}}}},\nonumber\\
\mathrm{s.t.}& \sum_{k=1}^{K}\mathrm{tr}\left ( \mathbf{D}_{m}\mathbf{W}_{k}\mathbf{D}^{H}_{m} \right ) \leq  P_{\mathrm{t}}, \label{p2.a.1}\\
&\mathbf{W}_{k} \succeq \mathbf{0} \, \forall k, \, \mathrm{rank}(\mathbf{W}_{k}) = 1.\label{p2.a.2}\\
&\color{black} \eqref{compschur1}, \eqref{compschur2},\eqref{p2.sinr},\eqref{compleakage}\label{p2.a.3}
% & \eqref{p2.1.c1}-\eqref{p2.1.c4} \label{p2.1.1.a.1};\\
\end{align}
\end{subequations}
\begin{subequations}
\begin{align}
\mathrm{(P2.B):}& \,\, \underset{\{\mathbf{W}_{k}\}}{\mathrm{minimize}}\quad I_{\mathrm{comp}},\nonumber\\
\mathrm{s.t.}&\color{black}\begin{bmatrix}
\mathbf{F}_{m,\mathrm{comp}}& \mathbf{I}_{:l}\\ \mathbf{I}^{T}_{:l}  
 & t^{*l}_{m,\mathrm{comp}} \\
\end{bmatrix} \succeq 0 \quad \forall l\in\{1,2,3\} \quad \forall m\label{p2.b.1}\\
&\color{black}t^{1}_{m,\mathrm{comp}}\leq f^{*}, \quad \forall m,\label{p2.b.2}\\
& \mathrm{tr}\left (  \mathbf{Q}_{k}  \mathbf{W}_{k} \right )-\gamma^{*} I_{\mathrm{comp}}\geq \gamma \sigma^{2}_{\mathrm{C}} \,\, \forall k,\label{p2.b.3}\\
&\eqref{compleakage}, \eqref{p2.a.1}, \eqref{p2.a.2}, .
\end{align}
\end{subequations}
\color{black}Problem (P2.A) maximizes the objective function of (P2) for a given maximum value of the co-channel interference $I_{\mathrm{comp}}$. Equation \eqref{p2.a.1} is the equivalent representation of \eqref{p2.c3}. Problem (P2.B) minimizes the co-channel interference while guaranteeing a maximum CRB $f^{*}$ and minimum SINR $\gamma^{*}$ through \eqref{p2.b.1}-\eqref{p2.b.3}. \color{black}Dropping the rank constraints, (P2.A) and (P2.B) are convex optimization problems. Algorithm 1 summarizes the robust block-level precoder design procedure. 

\begin{algorithm}[]
\caption{Robust Block Level Precoder Design}\label{algo1}
%\SetAlgoLined
\textbf{Input}: $u$, $\{\mathbf{h}_{m,nk}\}$, $\{\mathbf{h}_{k}\}$, $\{\mathbf{G}_{mn}\} $, $I^{\mathrm{intra}}_{m,\mathrm{cbf}}$, $I^{\mathrm{inter}}_{\mathrm{cbf}}$,$\{\mathbf{W}_{mk}\}$, $P_{\mathrm{leak}}$, $I_{\mathrm{comp}}$;\\
\While{no convergence}
{
\If {BSs in CBF mode}{
Determine $\mathbf{C}_{m,\mathrm{cbf}}$ using $\{\mathbf{W}_{mk}\}$ in \eqref{cmcbf};\\
\color{black}Solve (P1.1.A) for each BS to obtain $\gamma^{*}$ and $\{{t^{*l}_{m,\mathrm{cbf}}}\}$;\\
Solve (P1.1.B) to update $I^{\mathrm{intra}}_{m,\mathrm{cbf}}$, $I^{\mathrm{inter}}_{\mathrm{cbf}}$,$P_{\mathrm{leak}}$, $\{\mathbf{W}_{mk}\}$: $I^{\mathrm{inter}}_{\mathrm{cbf}}= \mathrm{max}\left(\{I^{\mathrm{inter}}_{m,\mathrm{cbf}}\}\right)$ \\
% $I^{\mathrm{inter}}_{\mathrm{cbf}}$ = $\mathrm{maximum}\left(I^{\mathrm{inter}}_{m,\mathrm{cbf}}\right)$;
}
\If{BSs in CoMP mode}{
Solve (P2.1.A) to obtain $\gamma^{*}$ and ${f^{*}}$;\\
Solve (P2.1.B) to update $I_{\mathrm{comp}}$: $I_{\mathrm{comp}}= \mathrm{max}\left(\{I_{m,\mathrm{comp}}\}\right)$\\
}
}
\textbf{Output}:{$\{\mathbf{W}_{m,k}\}$, $\{\mathbf{W}_{k}\}$}.
\end{algorithm} 
\color{black}Note that, for a given $\gamma$, $\rho=1$, $\delta=0$ $J=1$ setting, (P1), and (P2) are equivalent to the BLP design formulation, Eq. (18), of \cite{liu2021cramer}. We verify the equivalence of the corresponding CRB expressions in Fig. \ref{resultfigure1}. Certainly, our investigation into the multi-cell version of the CRB minimization problem introduces substantial differences in modeling. These distinctions make it impractical to directly apply the solutions presented in \cite{liu2021cramer}\color{black}. Furthermore, when $\delta=0$, both (P1) and (P2)  reduce to the respective CBF and CoMP BLP design problems considered in our previous work \cite{babu2023multi}.
\section{Robust Symbol Level Precoding}\label{SLP}
\begin{figure}[]
\centering
\captionsetup{justification=centering}
\centerline{\includegraphics[width=0.6\columnwidth]{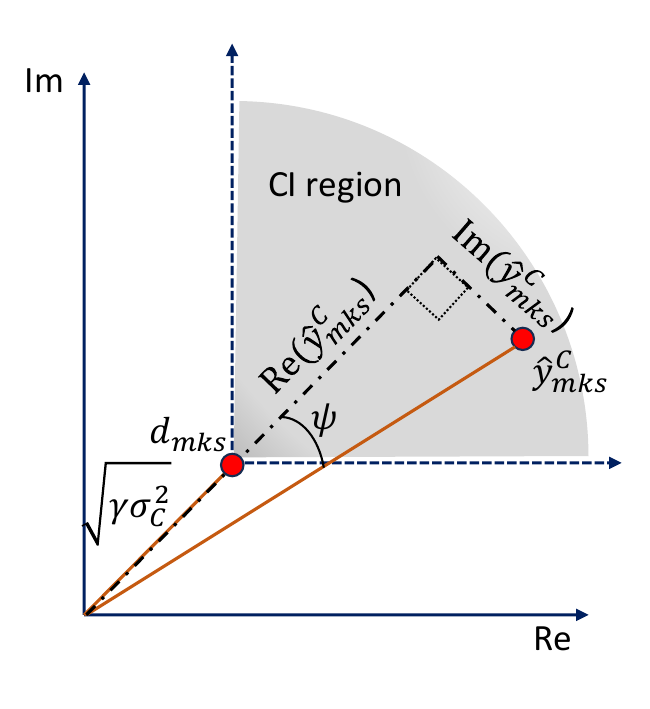}}
\caption{SLP: the $\langle d_{mks}$ rotated noiseless received signal $\hat{y}_{mks}=\mathbf{\tilde{h}}^{T}_{m,mks} \mathbf{{x}}_{ms}d_{mks}$ should fall in the CI region of the transmitted QPSK symbol $d_{mks}$.}
\label{figureslp}
\end{figure}
In this section, we explain the SLP design framework for CBF and CoMP BS coordination schemes.
\subsection{\color{black}Overview: SLP}
\color{black}The main idea involves leveraging co-channel interference constructively to increase the received signal power. This is achieved by instantaneously aligning the interfering signals with the desired signal at each receive antenna, as shown in Fig. \ref{figureslp}. \color{black}For a single-cell system, if the $s^{\mathrm{th}}$ transmitted symbol is M-PSK-modulated: $d_{mks}=d e^{j\phi_{mks}}$, the received signal at $U_{mk}$ can be represented as 
\begin{align}
{y}_{mks}^{\mathrm{C}} &= \mathbf{\tilde{h}}^{T}_{m,mk}\sum_{l=1}^K\mathbf{w}_{ml}d_{mls}+ {z}_{mks}^{\mathrm{C}},\nonumber\\
&= \mathbf{\tilde{h}}^{T}_{m,mk}\sum_{l=1}^K\mathbf{w}_{ml}e^{j(\phi_{mls}-\phi_{mks})}d_{mks}+ {z}_{mks}^{\mathrm{C}},\nonumber\\
&= \mathbf{\tilde{h}}^{T}_{m,mks} \mathbf{{x}}_{ms}d_{mks} +{z}_{mks}^{\mathrm{C}},
\label{yslp} 
\end{align}
where $\mathbf{\tilde{h}}^{T}_{m,mks}=\mathbf{\tilde{h}}^{T}_{m,mk} e^{j(-\phi_{mks})}$ and $\mathbf{{x}}_{ms}=\sum_{l=1}^K\mathbf{w}_{ml}e^{j(\phi_{mls})}$ (for $d=1$).
Hence, from Fig. \ref{figureslp}, the condition for $\mathbf{y}_{mks}^{\mathrm{C}}$ lying in the CI region of $d_{mks}$ to achieve a SINR of $\gamma$ can be expressed as,
\begin{align}
|\mathrm{Im}\left(\mathbf{\hat{h}}^{T}_{m,mks} \mathbf{{x}}_{ms}\right)| \leq \left(\mathrm{Re}\left(\mathbf{\tilde{h}}^{T}_{m,mks} \mathbf{{x}}_{ms}-\sigma_{\mathrm{C}}\sqrt{\gamma}\right)\right)\mathrm{tan}\psi, \label{slpcondition}
\end{align}
where, $\psi=\pi/M_{\mathrm{psk}}$.
For a detailed discussion about SLP, we refer the readers to \cite{masouros2015exploiting} and \cite{chrisSymbol}.
\subsection{SLP: Coordinated Beamforming}
In the CBF case, the $s^{\mathrm{th}}$ received symbol at $U_{mk}$ is obtained by extending \eqref{yslp} to a multi-cell scenario as follows: 
\begin{align}
{y}_{mks}^{\mathrm{C}} &= \mathbf{\tilde{h}}^{T}_{m,mks} \mathbf{x}_{ms}d_{mks} +\sum_{n\neq m}^{J}\mathbf{\tilde{h}}^{T}_{n,mks} \mathbf{x}_{ns}d_{nks}+\mathbf{z}_{mks}^{\mathrm{C}},
\label{yslpcbf} 
\end{align}
 where $\mathbf{\tilde{h}}^{T}_{n,mks}= \mathbf{\tilde{h}}^{T}_{n,mk} e^{j(-\phi_{nks})}$ and $\mathbf{x}_{ns}=\sum_{l=1}^K\mathbf{w}_{nl}e^{j(\phi_{nls})}$. Since the data symbols are not shared among the BSs in the CBF mode, the SLP design can only align the intra-cell interference to enhance the received signal power, whereas the inter-cell interference degrades the SINR value. Hence, the SINR constraint is equivalently written as a CI constraint given by \eqref{slpconditioncbf}, where $\gamma^{'}=\sqrt{\gamma}$.
 
 The corresponding problem formulation using SLP is given by 
 \begin{subequations}
 \begin{align}
\mathrm{(P3):}& \color{black}\underset{\{\mathbf{x}_{ms}\}, \{\mathbf{R}_{\mathbf{X}_{ms}}\}}{\mathrm{maximize}}\,\,\,\,  \frac{{-t^{1}_{m,\mathrm{cbf}}}\,\rho}{\mathrm{NF}^{\mathrm{slp}}_{\mathrm{R},{\mathrm{cbf}}}} \,\,  +\frac{(1-\rho)\gamma^{'}}{\mathrm{NF}^{\mathrm{slp}}_{\mathrm{C},{\mathrm{cbf}}}},\nonumber\\
\mathrm{s.t.}&\color{black}\begin{bmatrix}
\mathbf{F}_{m,\mathrm{cbf}}& \mathbf{I}_{:l}\\ \mathbf{I}^{T}_{:l}  
 & t^{l}_{m,\mathrm{cbf}} \\
\end{bmatrix} \succeq 0 \quad \forall l\in\{1,2,3\} \quad \forall m\label{slcbfschur}\\
& \mathbf{v}^{\mathrm{T}}_{mn} \frac{1}{L}\sum_{s}^{L} \mathbf{R}_{\mathbf{X}_{ms}} \mathbf{v}^{*}_{mn} \leq P_{\mathrm{leak}}, \quad \forall n, \label{leakageslp}\\
& \frac{1}{L}\mathrm{tr}\left(\sum_{s=1}^{L}\mathbf{R}_{\mathbf{X}_{ms}}\right)  \leq P_{\mathrm{t}},\label{p3.c3}\\
&\begin{bmatrix}
 \mathbf{R}_{\mathbf{X}_{ms}}& \mathbf{x}_{ms}\\ 
 \mathbf{x}^{H}_{ms}& 1
\end{bmatrix}\succeq 0; \,\, \mathbf{R}_{\mathbf{X}_{ms}}; \,\, \succeq 0 \, \forall s,\label{p3.c2}\\
&\eqref{slpconditioncbf}.
\end{align}
\end{subequations}
\color{black}The objective function of (P3) is the weighted combination of minimizing the CRB and maximizing the minimum SINR values. Note that, unlike the BLP design problem (P1), the optimization variables in this context are the transmit symbol vector $\mathbf{x}_{ms}$ and its covariance matrix $\mathbf{R}_{\mathbf{x}_{ms}}$. \color{black}Additionally, the entries of $\mathbf{F}_{m,{\mathrm{cbf}}}$ value is a function $\mathbf{R}_{\mathbf{x}_{m}} =(1/L)\sum_{s=1}^{L}\mathbf{x}_{ms}\mathbf{x}^{H}_{ms}=(1/L)\sum_{s=1}^{L}\mathbf{R}_{\mathbf{x}_{ms}}$. Similar to (P1), $\mathrm{NF}^{\mathrm{slp}}_{\mathrm{R},{\mathrm{cbf}}}$ and $\mathrm{NF}^{\mathrm{slp}}_{\mathrm{C},{\mathrm{cbf}}}$ are determined by solving (P3) for $\rho=1$ and $\rho=0$, respectively. Equation \eqref{slcbfschur} limits the $l^{\mathrm{th}}$ diagonal of $\mathbf{F}^{-1}_{m,\mathrm{cbf}}$ to be less than or equal to $t^{l}_{m,\mathrm{cbf}}$ using the Schur complement. Equations \eqref{leakageslp} and \eqref{p3.c3} are the respective inter-cell power leakage and average power constraints. Equation \eqref{p3.c2} is the semi-definite relaxed (SDR) representation of the relation between $\mathbf{R}_{\mathbf{x}_{ms}}$ and ${\mathbf{x}_{ms}}$ using the Schur complement. Solving (P3) poses challenges stemming from (a) the infinite possibilities of $\mathbf{e}_{mks}$ in \eqref{slpconditioncbf}, (b) the non-convex multiplication between $\gamma^{'}$ and the interference terms in \eqref{slpconditioncbf} and (c) the non-convex form of $\mathbf{C}_{m,\mathrm{cbf}}$ in \eqref{slcbfschur}.\color{black}
\begin{figure*}
\begin{subequations}
\begin{align}
&\underset{\|\mathbf{e}_{mks}\|^{2}\leq \delta^{2}}{\mathrm{max}}\left(\left|\mathrm{Im}\left\lbrace\mathbf{\tilde{h}}^{T}_{m,mks}\mathbf{x}_{ms} \right\rbrace\right| - \left(\mathrm{Re}\left\lbrace\mathbf{\tilde{h}}^{T}_{m,mks}\mathbf{x}_{ms} \right\rbrace\right)\mathrm{tan}\psi +\gamma^{'}\left(\sqrt{\left(\sigma^{2}_{\mathrm{C}}+ \sum_{n\neq m}^{J}|\mathbf{\tilde{h}}^{T}_{n,mks}\mathbf{x}_{ns}|^{2}\right)}\right)\mathrm{tan}\psi\right) \leq 0 \, \forall k,s\label{slpconditioncbf}\\
& \underset{\|\mathbf{e}_{mks}\|^{2}\leq \delta^{2}}{\mathrm{max}}\left(\mathbf{\hat{f}}^{T}_{m,mks}\mathbf{\hat{x}}_{ms}-\mathbf{\hat{f}}^{T}_{m,mks}\mathbf{\hat{x}}^{'}_{ms}\mathrm{tan}\psi  + \gamma^{'} \left(\sqrt{\left(\sigma^{2}_{\mathrm{C}}+ \sum_{n\neq m}^{J}|\mathbf{\tilde{h}}^{T}_{n,mks}\mathbf{x}_{ns}|^{2}\right)}\right) \mathrm{tan}\psi + \mathbf{\hat{e}}_{mks} (\mathbf{\hat{x}}_{ms}-\mathbf{\hat{x}}^{'}_{ms}\mathrm{tan}\psi) \right)\leq 0\label{ci1}\\
& \underset{\|\mathbf{e}_{mks}\|^{2}\leq \delta^{2}}{\mathrm{max}}\left(-\mathbf{\hat{f}}^{T}_{m,mks}\mathbf{\hat{x}}_{ms}-\mathbf{\hat{f}}^{T}_{m,mks}\mathbf{\hat{x}}^{'}_{ms}\mathrm{tan}\psi  + \gamma^{'} \left(\sqrt{\left(\sigma^{2}_{\mathrm{C}}+ \sum_{n\neq m}^{J}|\mathbf{\tilde{h}}^{T}_{n,mks}\mathbf{x}_{ns}|^{2}\right)}\right) \mathrm{tan}\psi + \mathbf{\hat{e}}_{mks} (-\mathbf{\hat{x}}_{ms}-\mathbf{\hat{x}}^{'}_{ms}\mathrm{tan}\psi)\right) \leq 0\label{ci2}
\end{align}
\end{subequations}
\end{figure*}

To tackle (a) and (b), let $\mathbf{\tilde{h}}_{m,mks}=\mathbf{\tilde{h}}^{\mathrm{R}}_{m,mks} +j\mathbf{\tilde{h}}^{\mathrm{I}}_{m,mks}=\mathbf{{h}}^{\mathrm{R}}_{m,mks} +j\mathbf{{h}}^{\mathrm{I}}_{m,mks}+\mathbf{e}^{\mathrm{R}}_{mks} +j\mathbf{e}^{\mathrm{I}}_{mks}$, then 
\begin{align}
    &\mathrm{Im}\left\lbrace\mathbf{\tilde{h}}^{T}_{m,mks}\mathbf{x}_{ms} \right\rbrace = \mathbf{\hat{f}}^{T}_{m,mks}\mathbf{\hat{x}}_{ms} +  \mathbf{\hat{e}}^{T}_{mks}\mathbf{\hat{x}}_{ms}\\
    & \mathrm{Re}\left\lbrace\mathbf{\tilde{h}}^{T}_{m,mks}\mathbf{x}_{ms} \right\rbrace = \mathbf{\hat{f}}^{T}_{m,mks}\mathbf{\hat{x}}^{'}_{ms} +  \mathbf{\hat{e}}^{T}_{mks}\mathbf{\hat{x}}^{'}_{ms}
\end{align}
where $ \mathbf{\hat{f}}_{m,mks}=\left[\mathbf{{h}}^{\mathrm{R}}_{m,mks};\mathbf{{h}}^{\mathrm{I}}_{m,mks}\right]$, $ \mathbf{\hat{e}}_{mks}=\left[\mathbf{{e}}^{\mathrm{R}}_{mks};\mathbf{{e}}^{\mathrm{I}}_{mks}\right]$, $\mathbf{\hat{x}}_{ms}=\left[\mathbf{x}^{\mathrm{I}}_{ms};\mathbf{x}^{\mathrm{R}}_{ms}\right]$, and $\mathbf{\hat{x}}^{'}_{ms}=\left[\mathbf{x}^{\mathrm{R}}_{ms};-\mathbf{x}^{\mathrm{I}}_{ms}\right]$. \color{black}Then \eqref{slpconditioncbf} can be equivalently expressed as \eqref{ci1} and \eqref{ci2}. Taking the maximum value of the LHS, the constraints are equal to,\color{black}
\begin{align}
    & \mathbf{\hat{f}}^{T}_{m,mks}\mathbf{\hat{x}}_{ms}-\mathbf{\hat{f}}^{T}_{m,mks}\mathbf{\hat{x}}^{'}_{ms}\mathrm{tan}\psi  + \gamma^{'} \sqrt{\sigma^{2}_{\mathrm{C}}+I^{^{2}\mathrm{inter}}_{\mathrm{sl,cbf}}} \mathrm{tan}\psi \nonumber\\ 
    &+ \delta \|\mathbf{\hat{x}}_{ms}-\mathbf{\hat{x}}^{'}_{ms}\mathrm{tan}\psi\| \leq 0,\label{f1}\\
& -\mathbf{\hat{f}}^{T}_{m,mks}\mathbf{\hat{x}}_{ms}-\mathbf{\hat{f}}^{T}_{m,mks}\mathbf{\hat{x}}^{'}_{ms}\mathrm{tan}\psi  + \gamma^{'} \sqrt{\sigma^{2}_{\mathrm{C}}+I^{^{2}\mathrm{inter}}_{\mathrm{sl,cbf}}} \mathrm{tan}\psi \nonumber\\
&+ \delta \|\mathbf{\hat{x}}_{ms}+\mathbf{\hat{x}}^{'}_{ms}\mathrm{tan}\psi\| \leq 0, \label{f2}
\end{align}
where $I^{^{2}\mathrm{inter}}_{\mathrm{sl,cbf}}$ is the maximum ICI that satisfies
\begin{align}
  & \underset{\|\mathbf{e}_{nks}\|^{2}\leq \delta^{2}}{\mathrm{max}} \|\mathbf{\tilde{h}}^{T}_{m,nks}\mathbf{x}_{ms}\|^{2} \leq  {\frac{I^{^{2}\mathrm{inter}}_{\mathrm{sl,cbf}}}{J-1}}. \label{slpcbfici}
  \end{align}
  Note that the squared representation of ICI: $I^{^{2}\mathrm{inter}}_{\mathrm{sl,cbf}}$, is to ensure convexity for the constraints \eqref{f1} and \eqref{f2}. Using Cauchy-Schwarz inequality, we have
  \begin{align}
   &\left(\|\mathbf{{h}}^{T}_{m,nks}+\mathbf{{e}}^{T}_{nks}\right)\mathbf{x}_{ms}\|^{2}
    \leq \left(\|\mathbf{{h}}^{T}_{m,nks}\mathbf{x}_{ms}\|+\|\mathbf{{e}}^{T}_{nks}\mathbf{x}_{ms}\|\right)^{2}
    \end{align}
    Hence, \eqref{slpcbfici} is equivalently written as
    \begin{align}
%     % &\leq \|\mathbf{\hat{h}}^{T}_{m,nks}\|^{2}+\delta^2\|\mathbf{x}_{ms}\|^{2}\\
%     &\|\mathbf{\tilde{h}}^{T}_{m,nks}\mathbf{x}_{ms}\| \leq  \|\mathbf{{h}}^{T}_{m,nks}\mathbf{x}_{ms}\|+\delta\|\mathbf{x}_{ms}\|\leq \sqrt{\frac{I_{\mathrm{max}}}{J-1}}
% \end{align}
% \begin{align}
    &\|\mathbf{{h}}^{T}_{m,nks}\mathbf{x}_{ms}\|+\delta\|\mathbf{x}_{ms}\|\leq {\frac{I^{\mathrm{inter}}_{\mathrm{sl,cbf}}}{\sqrt{J-1}}} \label{f3}.
\end{align}
Note that for a given $I^{\mathrm{inter}}_{\mathrm{sl},\mathrm{cbf}}$, the constraints \eqref{f1}, \eqref{f2}, and \eqref{f3} that combined represent the SINR constraint are convex. \color{black}Finally, given $\{\mathbf{R}_{\mathbf{X}_{n}}\}$, the estimation of $\mathbf{C}_{m,\mathrm{cbf}}$ can be achieved using \eqref{cmcbf}. This makes the entries of $\mathbf{F}_{m,\mathrm{cbf}}$ in \eqref{slcbfschur} an affine function of $\mathbf{R}_{\mathbf{X}_{m}}$.\color{black} 

Hence, we solve (P3) by solving the following two optimization problems alternatively untill convergence: 
\begin{subequations}
\begin{align}
\mathrm{(P3.A):}& \color{black}\underset{\{\mathbf{x}_{ms}\}, \{\mathbf{R}_{\mathbf{X}_{ms}}\}}{\mathrm{maximize}}\,\,\,\,  \frac{{-t^{1}_{m,\mathrm{cbf}}}\,\rho}{\mathrm{NF}^{\mathrm{slp}}_{\mathrm{R},{\mathrm{cbf}}}} \,\,  +\frac{(1-\rho)\gamma^{'}}{\mathrm{NF}^{\mathrm{slp}}_{\mathrm{C},{\mathrm{cbf}}}}, \nonumber\\
\mathrm{s.t.} & \eqref{slcbfschur}-\eqref{p3.c2}, \eqref{f1}, \eqref{f2}, \eqref{f3}; \label{p3.a.1}
\end{align}
\end{subequations}
%  and for an obtained solution of (P3.A): $\gamma=\gamma^{*}$ and ${\mathrm{FIV}^{*^\mathrm{scbf}}_{mm}}$,
\begin{subequations}
 \begin{align}
\mathrm{(P3.B):} &\underset{\{\mathbf{x}_{ms}\},\{\mathbf{R}_{\mathbf{x}_{ms}}\}}{\mathrm{minimize}}\,\,\,\,\frac{P_{\mathrm{leak}}}{P^{\mathrm{max}}_{\mathrm{leak}}} + \frac{I^{{2}\mathrm{inter}}_{m,\mathrm{sl,cbf}}}{I_{\mathrm{max}}} \nonumber\\
\mathrm{s.t.}&\color{black}\begin{bmatrix}
\mathbf{F}_{m,\mathrm{cbf}}& \mathbf{I}_{:l}\\ \mathbf{I}^{T}_{:l}  
 & t^{*l}_{m,\mathrm{cbf}} \\
\end{bmatrix} \succeq 0 \quad \forall l\in\{1,2,3\} \quad \forall m,\label{p3.b.1}\\
& \mathbf{\hat{f}}^{T}_{m,mks}\mathbf{\hat{x}}_{ms}-\mathbf{\hat{f}}^{T}_{m,mks}\mathbf{\hat{x}}^{'}_{ms}\mathrm{tan}\psi\nonumber\\ & + \gamma^{'^{*}} \sqrt{\sigma^{2}_{\mathrm{C}}+I^{^{2}\mathrm{inter}}_{\mathrm{sl,cbf}}} \mathrm{tan}\psi  
    + \delta \|\mathbf{\hat{x}}_{ms}-\mathbf{\hat{x}}^{'}_{ms}\mathrm{tan}\psi\| \leq 0\label{p3.b.2}\\
& -\mathbf{\hat{f}}^{T}_{m,mks}\mathbf{\hat{x}}_{ms}-\mathbf{\hat{f}}^{T}_{m,mks}\mathbf{\hat{x}}^{'}_{ms}\mathrm{tan}\psi \nonumber\\
& + \gamma^{'^{*}} \sqrt{\sigma^{2}_{\mathrm{C}}+I^{^{2}\mathrm{inter}}_{\mathrm{sl,cbf}}} \mathrm{tan}\psi + \delta \|\mathbf{\hat{x}}_{ms}+\mathbf{\hat{x}}^{'}_{ms}\mathrm{tan}\psi\| \leq 0 \label{p3.b.3}\\
&\|\mathbf{{h}}^{T}_{m,nks}\mathbf{x}_{ms}\|+\delta\|\mathbf{x}_{ms}\|\leq {\frac{I^{\mathrm{inter}}_{m,\mathrm{sl,cbf}}}{\sqrt{J-1}}},\label{p3.b.4}\\
& I^{\mathrm{inter}}_{m,\mathrm{sl,cbf}}<= I^{\mathrm{inter}}_{\mathrm{sl,cbf}},\eqref{leakageslp}- \eqref{p3.c2}\label{p3.b.5}.
\end{align}
\end{subequations}
\color{black}Problem (P3.A) maximizes the objective function of (P3) for given $I^{\mathrm{inter}}_{\mathrm{sl, cbf}}$ and $P_{\mathrm{leak}}$ values. In (P3.B), we minimize the leakage power to the neighboring cells' users while guaranteeing given sensing ($t^{*l}_{m,\mathrm{cbf}}$) and communication performance ($\gamma^{'^{*}}$) through \eqref{p3.b.1}-\eqref{p3.b.4}. The value of the constant $I^{\mathrm{inter}}_{\mathrm{sl, cbf}}$ is given by the maximum of $\{I^{\mathrm{inter}}_{m,\mathrm{sl, cbf}}\}$. \color{black}The corresponding precoder design procedure is similar to Steps 8 and 9 of Algorithm 1.
\subsection{SLP: Coordinated Multipoint}
\color{black}In the CoMP mode, the $s^{\mathrm{th}}$ received symbol at $U_{k}$ is given by 
\begin{align}    
{y}_{ks}^{\mathrm{C}} &= \mathbf{\tilde{h}}^{T}_{ks} \mathbf{x}_{s}d_{ks} + \mathbf{z}_{ks}^{\mathrm{C}},
\label{yslpcomp}
\end{align}
 where $\mathbf{\tilde{h}}^{T}_{ks}= \mathbf{\tilde{h}}^{T}_{k} e^{j(-\phi_{ks})}$ and $\mathbf{x}_{s}=\sum_{l=1}^{JK}\mathbf{w}_{l}e^{j(\phi_{ls})}$.
Since the data is shared among the BSs, all the interfering signal power can be aligned with the useful signal power. Hence, the SINR constraint in the CoMP mode, derived from \eqref{f1} and \eqref{f2}, can be written as \color{black}
\begin{align}
    & \mathbf{\hat{f}}^{T}_{ks}\mathbf{\hat{x}}_{s}-\mathbf{\hat{f}}^{T}_{ks}\mathbf{\hat{x}}^{'}_{s}\mathrm{tan}\psi  + \gamma^{'} {\sigma_{\mathrm{C}}} \mathrm{tan}\psi + \delta \|\mathbf{\hat{x}}_{s}-\mathbf{\hat{x}}^{'}_{s}\mathrm{tan}\psi\| \leq 0\label{f1comp},
\end{align}
\begin{align}
& -\mathbf{\hat{f}}^{T}_{ks}\mathbf{\hat{x}}_{s}-\mathbf{\hat{f}}^{T}_{ks}\mathbf{\hat{x}}^{'}_{s}\mathrm{tan}\psi  + \gamma^{'} {\sigma_{\mathrm{C}}} \mathrm{tan}\psi + \delta \|\mathbf{\hat{x}}_{s}+\mathbf{\hat{x}}^{'}_{s}\mathrm{tan}\psi\| \leq 0,\label{f2comp}
\end{align}
where $\mathbf{\tilde{h}}^{T}_{ks}=\mathbf{\tilde{h}}^{\mathrm{R}}_{ks} +j\mathbf{\tilde{h}}^{\mathrm{I}}_{ks}=\mathbf{{h}}^{\mathrm{R}}_{ks} +j\mathbf{{h}}^{\mathrm{I}}_{ks}+\mathbf{e}^{\mathrm{R}}_{ks} +j\mathbf{e}^{\mathrm{I}}_{ks}$ and $\mathbf{x}_{s}=\sum_{l=1}^{2K}\mathbf{w}_{l}e^{j(\phi_{ls})}$; $ \mathbf{\hat{f}}_{ks}=\left[\mathbf{{h}}^{\mathrm{R}}_{ks};\mathbf{{h}}^{\mathrm{I}}_{ks}\right]$, $ \mathbf{\hat{e}}_{ks}=\left[\mathbf{{e}}^{\mathrm{R}}_{ks};\mathbf{{e}}^{\mathrm{I}}_{ks}\right]$, $\mathbf{\hat{x}}_{s}=\left[\mathbf{x}^{\mathrm{I}}_{s};\mathbf{x}^{\mathrm{R}}_{s}\right]$, and $\mathbf{\hat{x}}^{'}_{s}=\left[\mathbf{x}^{\mathrm{R}}_{s};-\mathbf{x}^{\mathrm{I}}_{s}\right]$. The corresponding optimization problem is formulated as,
\begin{subequations}
\begin{align}
&\mathrm{(P4):} \color{black}\underset{\{\mathbf{x}_{s}\},\{\mathbf{R}_{\mathbf{x}_{s}}\}}{\mathrm{maximize}}\,\,\,\,  \frac{ -f\rho}{\mathrm{NF}^{\mathrm{slp}}_{\mathrm{R},{\mathrm{comp}}}} \,\,  +\frac{(1-\rho)\gamma^{'}}{\mathrm{NF}^{\mathrm{slp}}_{\mathrm{C},{\mathrm{comp}}}}, \nonumber\\
&\color{black}\begin{bmatrix}
\mathbf{F}_{m,\mathrm{comp}}& \mathbf{I}_{:l}\\ \mathbf{I}^{T}_{:l}  
 & t^{l}_{m,\mathrm{comp}} \\
\end{bmatrix} \succeq 0 \quad \forall l\in\{1,2,3\} \quad \forall m\label{p4.c1}\\
&\color{black}t^{1}_{m,\mathrm{comp}}\leq f, \quad \forall m,\label{p4.c2}\\
% & \gamma \geq \Gamma, \label{p3.c2}\\
&\begin{bmatrix}
 \mathbf{R}_{\mathbf{X}_{s}}& \mathbf{x}_{s}\\ 
 \mathbf{x}^{H}_{s}& 1
\end{bmatrix}\succeq 0 \,\, \forall s\,\, \mathbf{R}_{\mathbf{X}_{s}} \succeq 0,\label{p4.c3}\\
& \frac{1}{L}\mathrm{tr}\left(\mathbf{D}_{m}\sum_{s=1}^{L}\mathbf{R}_{\mathbf{X}_{s}}\mathbf{D}^{'}_{m}\right)  \leq P_{\mathrm{t}}, \label{p4.c4}\\
& \eqref{f1comp}, \eqref{f2comp}.
\end{align}
\end{subequations}
\color{black}Problem (P4) maximizes the same objective as (P2). Here, \eqref{p4.c1}, \eqref{p4.c2}, and \eqref{p4.c3} are the same as \eqref{compschur1}, \eqref{compschur2}, and \eqref{p2.c3}. The convex optimization problem (P4) can be solved using Matlab's CVX \cite{grant2014cvx}.
\color{black}Due to the SDR in (P3.A) and (P4): $\mathbf{R}_{\mathbf{X}_{s}} \geq \mathbf{x}_{s}\mathbf{x}^{H}_{s} $, the obtained objective values represent the upper bound. \color{black}Moreover, by setting $\rho=0$, (P3) and (P4) can be modified to the SLP design formulations, Eq. (16) and Eq. (39), in \cite{wei2019multi} where the objective was to minimize the total power subject to a minimum SINR constraint\color{black}.
\section{\color{black}Convergence, Complexity, and Overhead Analysis}
\color{black}As explained in Sections \ref{BLP} and \ref{SLP}, the robust precoders for CBF-BLP, CoMP-BLP, and CBF-SLP are obtained using the AO algorithm given in Algorithm 1. Recall the AO algorithm employed alternates between optimizing precoders to maximize the weighted combination of sensing and communication performance metrics under specified maximum interference values, and minimizing communication and sensing interference values while ensuring given sensing and communication performance thresholds. The convergence of  such an iterative algorithm is proven in \cite{bezdek2003convergence}, and we also show it \color{black}empirically \color{black} in Fig. \ref{resultfigure6}.

The optimization problems that give the robust precoders: (P1.1.A), (P1.1.B), (P2.1.A), (P2.1.B), (P3.A), (P3.B), and (P4) are convex optimization problems obtained through SDR technique which can be solved using Algorithm 1. Since there is no closed-form expression for the iterative complexity of SDP optimization, we give the time complexity of each iteration of Algorithm 1 solved using interior point methods. If $V$ is the number of variables, and $C$ is the number of constraints, then the computational complexity in time is 
\begin{align}
   \color{black} C_{\mathrm{t}}\sim \mathcal{O}(\sqrt{V}(CV^2+C^{2.7}+V^{2.7})\mathrm{log}(1/\epsilon)),
\end{align} a polynomial function of $V$ and $C$  with $\epsilon$ being the relative accuracy \cite{jiang2020faster}. Note that BLP optimizations ((P1.1.A), (P1.1.B), (P2.1.A), (P2.1.B)) have low computational complexity compared to the SLP optimization ((P3.A), (P3.B), and (P4)) since the number of constraints of the later increases with the length of the communication frame. However, the receiver complexity at the user end in the SLP cases is low since it leverages the co-channel interference constructively. Furthermore, the complexity and overhead of SLP can be reduced by adopting a block-level approach as proposed in \cite{li2021symbol}.

The total overhead of sharing CSI, user data and ICR direction amongst $J$ BSs in CoMP mode is 
\begin{align}
   \color{black} O_{\mathrm{CoMP}}\sim \mathcal{O}(J(J-1)(K(J \beta_{\mathrm{h}}+\beta_{\mathrm{S}})+\beta_{\mathrm{dir}}),
\end{align}
where $\beta_{\mathrm{h}}$ is the number of bits required to represent CSI from a BS to a user, $\beta_{\mathrm{S}}$ and $\beta_{\mathrm{dir}}$ represent bits required to exchange user's data and an ICR direction. The overhead associated with the CBF mode is 
\begin{align}
   \color{black} O_{\mathrm{CoMP}}\sim \mathcal{O}(J(J-1)(K(J \beta_{\mathrm{h}})+\beta_{\mathrm{dir}}),
\end{align}
and is relatively lower due to the non-exchange of user data amongst BSs.
\color{black}
\begin{figure}[]
\centering
\captionsetup{justification=centering}
\centerline{\includegraphics[width=0.9\columnwidth]{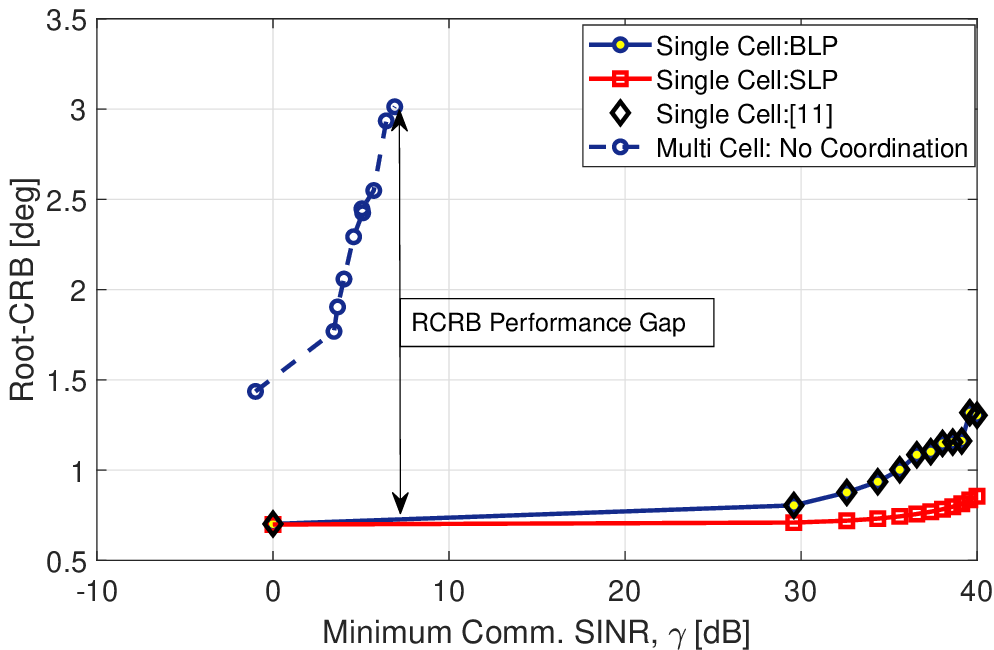}}
\caption{RCRB Vs minimum SINR performance when $N_{\mathrm{tx}}=10$, $N_{\mathrm{rx}}=10$, $K=3$, $\delta =0$, $M_{\mathrm{psk}}=4$. Single-cell solutions in the proximity of additional BSs are sub-optimal, motivating exploration of multi-cell ISAC setups.}
\label{resultfigure1}
\end{figure}
\begin{figure}[]
\centering
\captionsetup{justification=centering}
\centerline{\includegraphics[width=0.85\columnwidth]{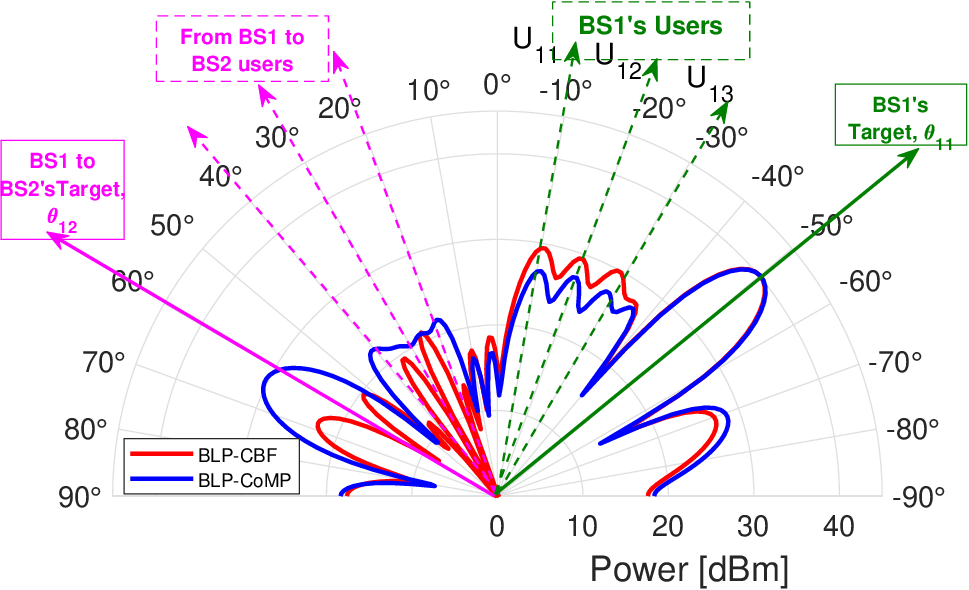}}
\caption{\color{black}Final beampattern for BS1, when $\gamma=40$ dB, $N_{\mathrm{tx}}=16$, $N_{\mathrm{rx}}=6$, $K=3$, $\delta =0$, $M_{\mathrm{psk}}=4$, $M_{\mathrm{psk}}=4$.}
\label{resultfigure2}
\end{figure}
\begin{figure}[]
\centering
\captionsetup{justification=centering}
\centerline{\includegraphics[width=0.8\columnwidth]{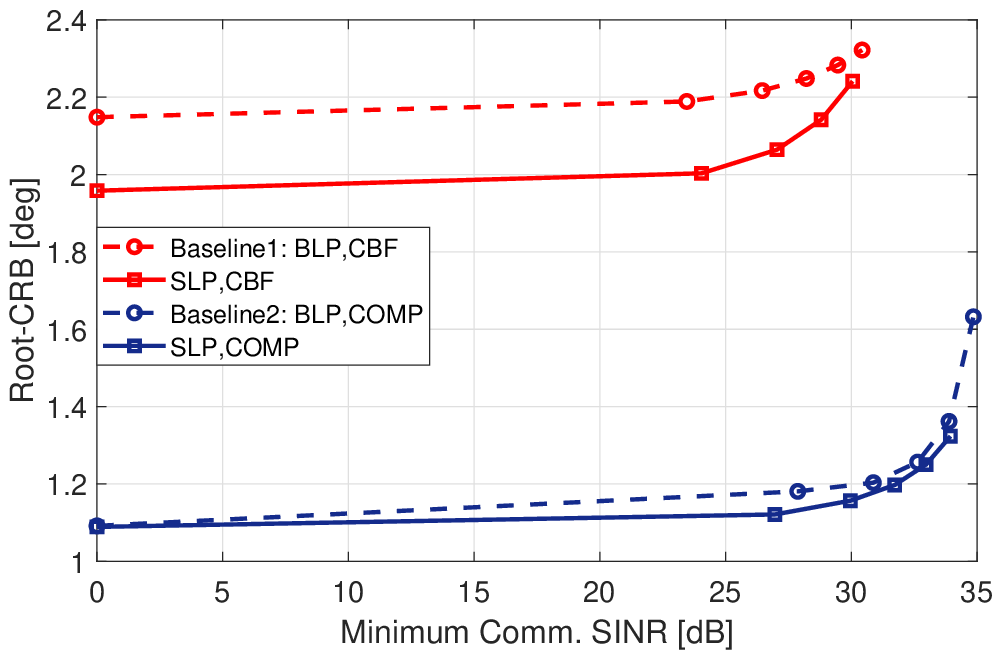}}
\caption{RCRB Vs. comm SINR threshold when $N_{\mathrm{tx}}=6$, $N_{\mathrm{rx}}=6$, $K=3$, $\delta=0$, $M_{\mathrm{psk}}=4$ .}
\label{resultfigure3}
\end{figure}
\begin{figure}[]
\centering
\captionsetup{justification=centering}
\centerline{\includegraphics[width=0.8\columnwidth]{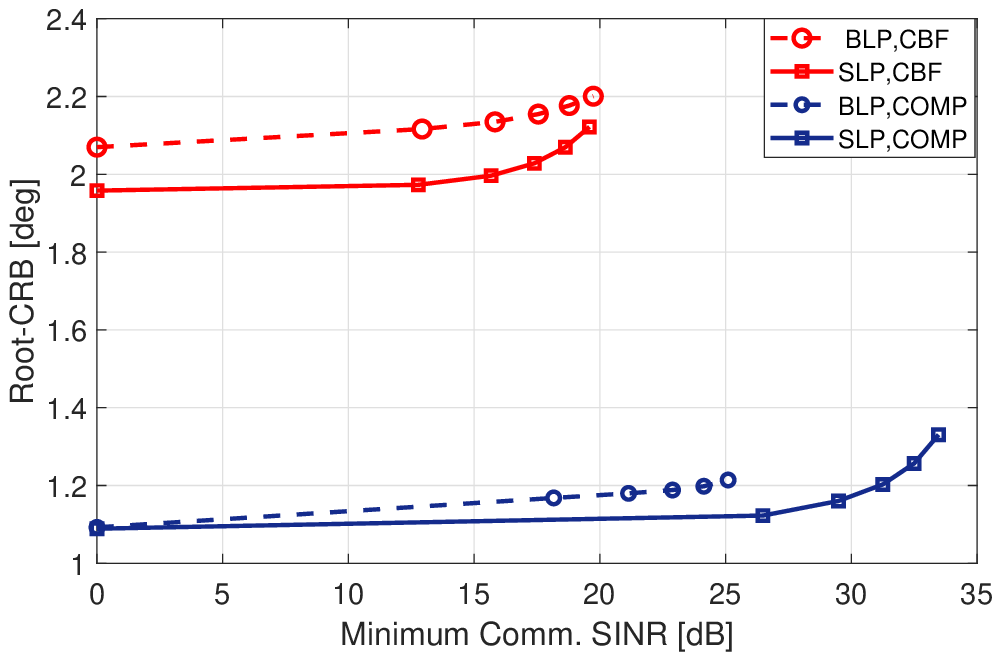}}
\caption{RCRB Vs. comm SINR threshold when, $N_{\mathrm{tx}}=6$, $N_{\mathrm{rx}}=6$, $K=3$, $\delta=0.01$, $M_{\mathrm{psk}}=4$ .}
\label{resultfigure4}
\end{figure}
\section{Numerical Evaluation}\label{numerical}
We present our main findings through numerical evaluations in this section. The simulation parameters are $J=2$, $K=3$, $P_{\mathrm{t}}/\sigma^{2}_{\mathrm{C}}=P_{\mathrm{t}}/\sigma^{2}_{\mathrm{R}}=40$ dB, $P_{\mathrm{leak}}^{\mathrm{max}}=4\pi P_{\mathrm{t}}$, and $f_c=5.6$ GHz. The targets are located at $\theta_{11}=-50^{\circ},\, \theta_{12}=60^{\circ},\,\theta_{22} = 50^{\circ},\, \theta_{21} = -60^{\circ}$.The communication channel gains are normalized such that the signal power received through an inter-cell link is considered to be reduced by a factor of 3. $\alpha_{mm}$ is selected such that the received  SNR  of  the  intra-cell echo  signal: $|\alpha_{mm}|^{2}LP_{\mathrm{t}}/\sigma^{2}_{\mathrm{R}}=1$ \cite{liu2021cramer}, whereas $|\alpha_{nm}|^{2}=|\alpha_{mm}|^{2}/3$. The ratio of the maximum tolerable interference values to the noise power is selected as $10$ dB. 
\begin{figure*} [t]  
        \centering
        \begin{subfigure}[b]{0.49\columnwidth}
            \centering
            \includegraphics[width=\columnwidth]{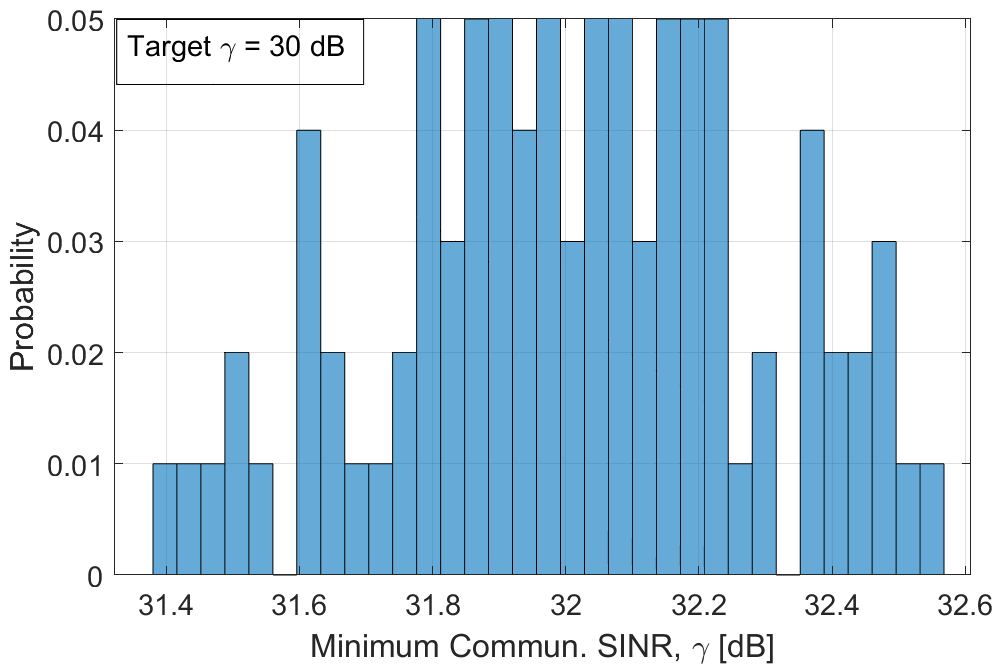}
            \caption{BLP CoMP.}\label{fig5a}
        \end{subfigure}
        \begin{subfigure}[b]{0.49\columnwidth}
            \centering
            \includegraphics[width=\columnwidth]{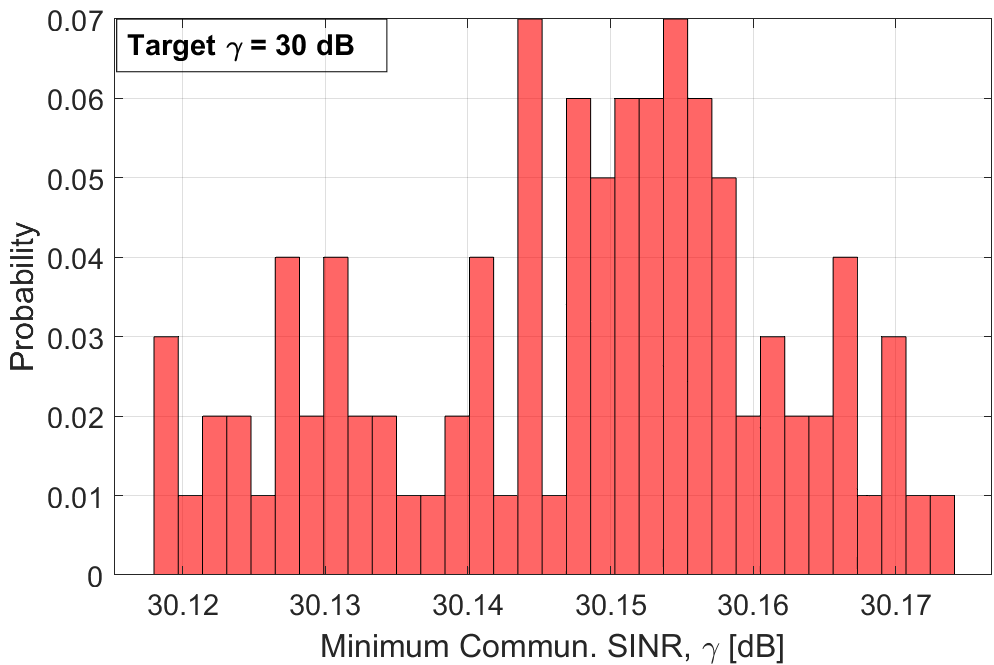}
            \caption{SLP CoMP.}\label{fig5b}
        \end{subfigure}
        \begin{subfigure}[b]{0.49\columnwidth}
            \centering
            \includegraphics[width=\columnwidth]{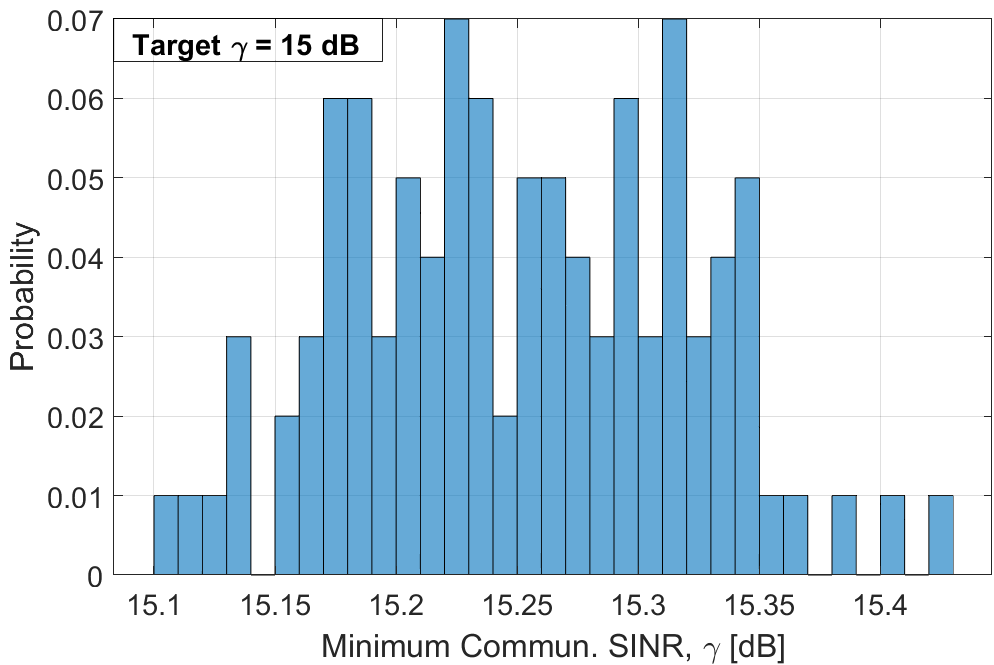}
            \caption{BLP CBF.}\label{fig5c}
        \end{subfigure}
        \begin{subfigure}[b]{0.49\columnwidth}
            \centering
            \includegraphics[width=\columnwidth]{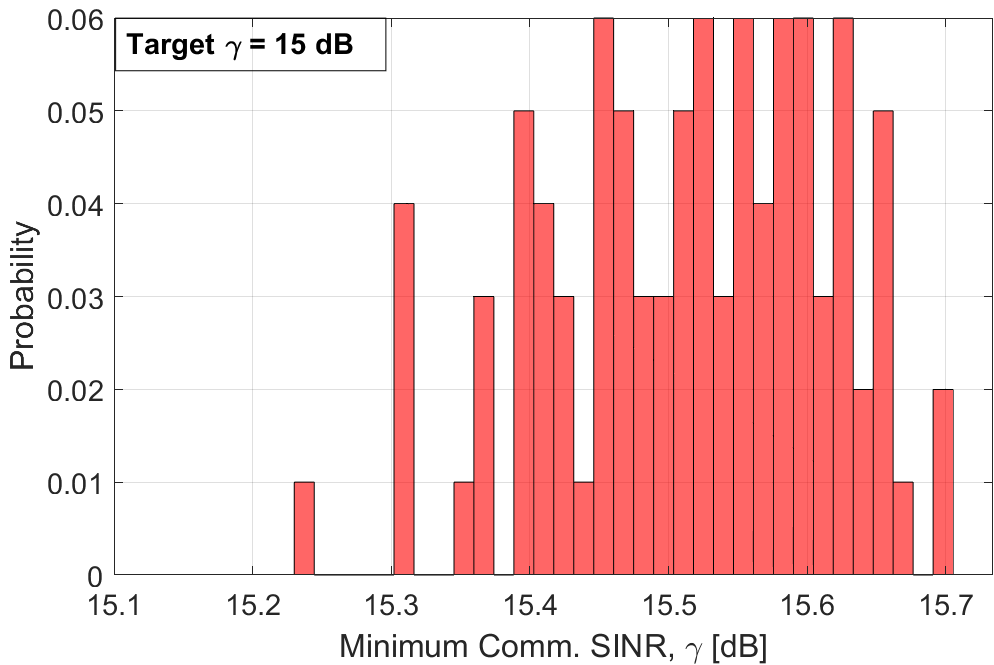}
            \caption{SLP CBF.}\label{fig5d}
        \end{subfigure}
        \caption{Probability distribution of observed minimum SINR for 100 different error vector realization, when $\gamma = 30$ dB (CoMP),$\gamma= 15$ dB (CBF) $N_{\mathrm{tx}}=16$, $N_{\mathrm{rx}}=6$, $\delta = 0.01$, $M_{\mathrm{psk}}=4$.} 
\end{figure*}

Fig. \ref{resultfigure1} shows the effect of neglecting the inter-cell links in a multi-cell ISAC system that uses BLP. The plots between the root CRB (RCRB) and the minimum SINR experienced by a communication user are obtained by solving (P1.1.A), neglecting the inter-cell links and assuming no CSI error. Then, the obtained precoders are used to determine the actual sensing versus communication performance trade-off. As seen in the figure, the additional signal power received through the inter-cell links can negatively affect both the sensing and the communication performances. This motivates us to consider a multi-cell ISAC setup since the single-cell solutions are suboptimal when additional BSs are deployed in proximity. \color{black}Moreover, in a single-cell setup, the derived CRB expression matches with that in \cite{liu2021cramer}\color{black}.
\begin{figure}[]
\centering
\captionsetup{justification=centering}
\centerline{\includegraphics[width=0.8\columnwidth]{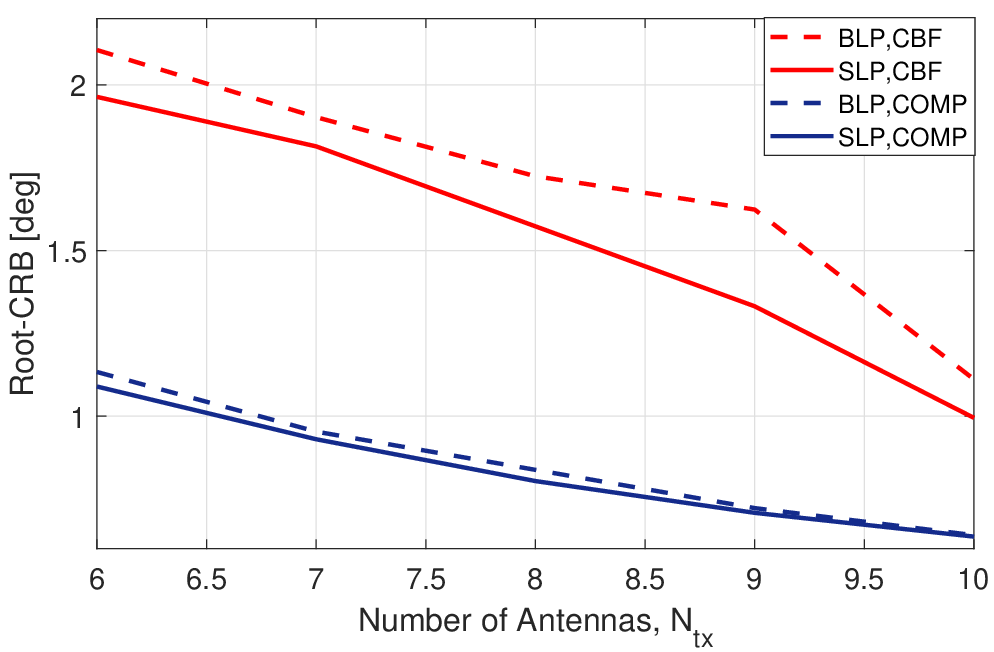}}
\caption{Variation of RCRB w.r.t the number of transmit antennas, when target $\gamma=10$ dB, $N_{\mathrm{rx}}=6$, $K=3$, $\delta=0.01$, $M_{\mathrm{psk}}=4$.}
\label{resultfigure7}
\end{figure}
\begin{figure}[]
\centering
\captionsetup{justification=centering}
\centerline{\includegraphics[width=0.8\columnwidth]{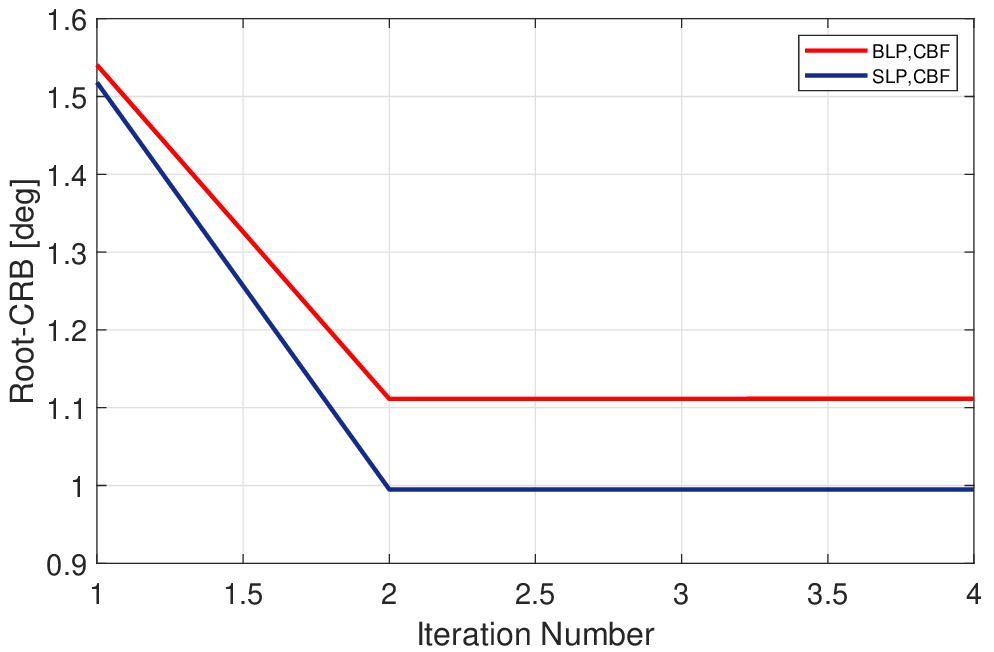}}
\caption{Convergence of Algorithm 1 when $\gamma=10$ dB, $N_{\mathrm{tx}}=10$, $N_{\mathrm{rx}}=6$, $K=3$, $\delta=0.01$, $M_{\mathrm{psk}}=4$.}
\label{resultfigure6}
\end{figure}

Fig. \ref{resultfigure2} shows the final transmit beampattern obtained using BLP applied to CBF and CoMP multi-cell scenarios. For better analysis, we assume the users are in LoS with the BSs along the directions marked in the figure. For the CBF scheme, the power radiated towards the neighboring BSs' \color{black}users/targets \color{black} through the sensing and the communication links is low compared to the CoMP scheme. This is because the inter-cell links degrade the sensing and communication performances in the CBF scheme due to the non-sharing of the user data among the BSs. On the other hand, the CoMP scenario improves the ISAC performance by radiating more power through the inter-cell links, \color{black}where it jointly serves the communication users while offering bi-static sensing\color{black}. 

Fig. \ref{resultfigure3} shows the RCRB- $\gamma$ trade-off for the considered four scenarios: BLP-CBF, BLP-CoMP, SLP-CBF, and SLP-CoMP, when the perfect knowledge of CSI is available at the BSs ($\delta=0$). \color{black}Recall that, in \cite{babu2023multi}, we examined BLP-CBF and BLP-CoMP under $\delta=0$. These solutions serve as benchmarks against their SLP counterparts\color{black}. In all the plots, the sensing error performance remains relatively low in the low $\gamma$ regime since the BS can meet the communication performance without sacrificing the power radiated towards the target. As the communication performance demand increases, the power radiated towards the target is reduced to meet the demand, thereby negatively affecting the sensing performance. As seen in the figure, the SLP technique outperforms the BLP counterparts: in the CBF scenario, the SLP technique utilizes the intra-cell interference constructively, thereby reducing the power needed to achieve the minimum SINR compared to the BLP technique, thus allowing the BS to radiate more power towards its target to improve the sensing performance. In the CoMP case, \color{black}firstly, there is an improvement in the sensing error due to the bi-static sensing. Moreover, \color{black} the inter-cell links are constructively utilized by the BLP technique to outperform the BLP-CBF scheme; however, it still suffers from the intra-cell interference while the SLP further improves the ISAC performance utilizing both the intra-cell and inter-cell interference. The slope of the plots changes as a function of the coupling between the users and the target channel. Here, we considered a weakly coupled scenario, whereas in a strongly coupled scenario where a user is a target of interest, both the sensing and communication performances can be improved in the high minimum SINR regime. 

Fig. \ref{resultfigure4} shows the effect of CSI uncertainty on the ISAC performance; the trend remains the same as in Fig. \ref{resultfigure3}. Moreover, the additional power set aside to incorporate the error in the available CSI information decreases the sensing accuracy and the range of achievable $\gamma$ compared to the perfect CSI available scenario. \color{black}The high-SINR regime characterizes the sensing performance achieved by precoder solutions primarily designed for communication, similar to the approach explored in \cite{wei2019multi}\color{black}.

In Figs. \ref{fig5a}-\ref{fig5d}, we show the robustness of the obtained solution against the CSI error when $\delta = 0.01$. The observed minimum SINR's probability distributions are obtained after 100 realizations of the error vector. As seen in the figure, all the observed minimum SINR values are above the set thresholds ($15$ dB for CBF and $30$ dB for CoMP), which shows that the minimum SINR constraint is satisfied under worst-case scenarios.

Fig. \ref{resultfigure7} shows the effect of the number of transmitting antennas on the sensing performance for a given minimum communication SINR of 10 dB. The broader beamwidth associated with a lower number of transmit antennas increases the interference to the users demanding more power to achieve the required SINR value. This causes a poorer sensing performance than with more antennas at the transmitting end. %Also, note that the performance gap between the BLP and SLP solutions decreases with increased antennas since the former could steer null towards other users with the available additional degrees of freedom, allowing to radiate more power towards the target direction. 

Fig. \ref{resultfigure6} shows the convergence of Algorithm 1 in terms of the while-loop iterations required. The RCRB value decreases as the iteration progresses, reaching a steady value in 2 iterations for the considered simulation parameters. The choice of $I^{\mathrm{intra}}_{m,\mathrm{cbf}}$, $I^{\mathrm{inter}}_{\mathrm{cbf}}$, $I_{\mathrm{comp}}$ plays an important role in the convergence speed and the final RCRB value: a higher value requires more number of iterations, whereas, a lower value restricts the complete usage of the available power budget to decrease the interference caused. This would result in a high RCRB value, hence a high target parameter estimation error. 
\section{Conclusion} \label{conclusion}
In this work, we considered robust precoding techniques for coordinated beamforming (CBF) and coordinated multipoint (CoMP) multi-cell scenarios of a multi-user multi-input-single-output (MISO) integrated sensing and communication (ISAC) networks. \color{black}We derived the Cramer-Rao bound (CRB) expressions of the considered multi-cell scenarios, whose inverse gives the lower bound on the estimation error variance. Using the derived expressions, we formulated precoder design problems that minimize CRB and maximize the minimum communication signal-to-noise-plus-interference ratio (SINR) subject to a total power constraint\color{black}. The design considered both the block-level and symbol-level precoding techniques, and the non-convexity of the problems were tackled by a combination of semidefinite relaxation (SDR) and alternating optimization (AO) techniques. As the first work to consider robust precoding techniques for a multi-cell ISAC system, we initially investigated the effect of additional signal power received through inter-cell links. As such, we deduced that if not carefully considered, the inter-cell sensing and communication links reduce the sensing accuracy while decreasing the received SINR. Among the considered multi-cell scenarios, the CoMP performs well in terms of both sensing accuracy and communication SINR, compared to the CBF scenario, since it uses the leaked power from the neighboring \color{black}BSs constructively for bi-static sensing and for boosting the serving users' SINRs\color{black}. Additionally, in both the multi-cell scenarios, the SLP solution outperformed the BLP counterpart due to its ability to use the co-channel interference constructively.
% \appendix
% let
% \begin{figure}[]
% \centering
% \captionsetup{justification=centering}
% \centerline{\includegraphics[width=0.8\columnwidth]{result_fig3_sigma_01.eps}}
% \caption{RCRB Vs minimum communication SINR; $(N_{\mathrm{t}},N_{\mathrm{r}})=(6,4)$, $\sigma=0.01$.}
% \label{resultfigure1}
% \end{figure}
\bibliographystyle{IEEEtran}
\bibliography{./bibliography.bib}

\begin{thebibliography}{10}
\providecommand{\url}[1]{#1}
\csname url@samestyle\endcsname
\providecommand{\newblock}{\relax}
\providecommand{\bibinfo}[2]{#2}
\providecommand{\BIBentrySTDinterwordspacing}{\spaceskip=0pt\relax}
\providecommand{\BIBentryALTinterwordstretchfactor}{4}
\providecommand{\BIBentryALTinterwordspacing}{\spaceskip=\fontdimen2\font plus
\BIBentryALTinterwordstretchfactor\fontdimen3\font minus \fontdimen4\font\relax}
\providecommand{\BIBforeignlanguage}[2]{{%
\expandafter\ifx\csname l@#1\endcsname\relax
\typeout{** WARNING: IEEEtran.bst: No hyphenation pattern has been}%
\typeout{** loaded for the language `#1'. Using the pattern for}%
\typeout{** the default language instead.}%
\else
\language=\csname l@#1\endcsname
\fi
#2}}
\providecommand{\BIBdecl}{\relax}
\BIBdecl

\bibitem{liu2020joint}
F.~Liu, C.~Masouros, A.~P. Petropulu, H.~Griffiths, and L.~Hanzo, ``{Joint Radar and Communication Design: Applications, State-of-the-Art, and the Road Ahead},'' \emph{IEEE Trans. on Commun.}, vol.~68, no.~6, pp. 3834--3862, 2020.

\bibitem{liu2022integrated}
F.~Liu, Y.~Cui, C.~Masouros, J.~Xu, T.~X. Han, Y.~C. Eldar, and S.~Buzzi, ``{Integrated sensing and communications: Toward dual-functional wireless networks for 6G and beyond},'' \emph{IEEE journal on selected areas in communications}, vol.~40, no.~6, pp. 1728--1767, 2022.

\bibitem{chen20235g}
W.~Chen, X.~Lin, J.~Lee, A.~Toskala, S.~Sun, C.~F. Chiasserini, and L.~Liu, ``5g-advanced toward 6g: Past, present, and future,'' \emph{IEEE Journal on Selected Areas in Communications}, vol.~41, no.~6, pp. 1592--1619, 2023.

\bibitem{union2022future}
I.~T. Union, ``{Future Technology Trends of Terrestrial International Mobile Telecommunications Systems Towards 2030 and Beyond},'' 2022.

\bibitem{kaushik2023towards}
A.~Kaushik, R.~Singh, S.~Dayarathna, R.~Senanayake, M.~Di~Renzo, M.~Dajer, H.~Ji, Y.~Kim, V.~Sciancalepore, A.~Zappone \emph{et~al.}, ``{Towards Integrated Sensing and Communications for 6G: A Standardization Perspective},'' \emph{arXiv preprint arXiv:2308.01227}, 2023.

\bibitem{9858656}
K.~Meng, Q.~Wu, S.~Ma, W.~Chen, K.~Wang, and J.~Li, ``{Throughput Maximization for UAV-Enabled Integrated Periodic Sensing and Communication},'' \emph{IEEE Trans. on Wireless Commun.}, vol.~22, no.~1, pp. 671--687, 2023.

\bibitem{liu2018toward}
F.~Liu, L.~Zhou, C.~Masouros, A.~Li, W.~Luo, and A.~Petropulu, ``{Toward Dual-Functional Radar-Communication Systems: Optimal Waveform Design},'' \emph{IEEE Trans. on Signal Processing}, vol.~66, no.~16, pp. 4264--4279, 2018.

\bibitem{5776640}
C.~Sturm and W.~Wiesbeck, ``{Waveform Design and Signal Processing Aspects for Fusion of Wireless Communications and Radar Sensing},'' \emph{Proceedings of the IEEE}, vol.~99, no.~7, pp. 1236--1259, 2011.

\bibitem{9723383}
D.~Ma, N.~Shlezinger, T.~Huang, Y.~Shavit, M.~Namer, Y.~Liu, and Y.~C. Eldar, ``A hardware prototype for joint radar-communication system using spatial modulation,'' in \emph{2021 55th Asilomar Conference on Signals, Systems, and Computers}, 2021, pp. 634--639.

\bibitem{temiz2023experimental}
M.~Temiz, C.~Horne, N.~J. Peters, M.~A. Ritchie, and C.~Masouros, ``{An Experimental Study of Radar-Centric Transmission for Integrated Sensing and Communications},'' \emph{IEEE Trans. on Microwave Theory and Techniques}, 2023.

\bibitem{liu2021cramer}
F.~Liu, Y.-F. Liu, A.~Li, C.~Masouros, and Y.~C. Eldar, ``{Cram{\'e}r-Rao Bound Optimization for Joint Radar-Communication Beamforming},'' \emph{IEEE Transactions on Signal Processing}, vol.~70, pp. 240--253, 2021.

\bibitem{9124713}
X.~Liu, T.~Huang, N.~Shlezinger, Y.~Liu, J.~Zhou, and Y.~C. Eldar, ``{Joint Transmit Beamforming for Multiuser MIMO Communications and MIMO Radar},'' \emph{IEEE Trans. on Signal Processing}, vol.~68, pp. 3929--3944, 2020.

\bibitem{9303435}
X.~Yuan, Z.~Feng, J.~A. Zhang, W.~Ni, R.~P. Liu, Z.~Wei, and C.~Xu, ``{Spatio-Temporal Power Optimization for MIMO Joint Communication and Radio Sensing Systems With Training Overhead},'' \emph{IEEE Transactions on Vehicular Technology}, vol.~70, no.~1, pp. 514--528, 2021.

\bibitem{7953658}
B.~Li and A.~P. Petropulu, ``{Joint Transmit Designs for Coexistence of MIMO Wireless Communications and Sparse Sensing Radars in Clutter},'' \emph{IEEE Trans. on Aerospace and Electronic Systems}, vol.~53, no.~6, pp. 2846--2864, 2017.

\bibitem{9385108}
M.~Temiz, E.~Alsusa, and M.~W. Baidas, ``{Optimized Precoders for Massive MIMO OFDM Dual Radar-Communication Systems},'' \emph{IEEE Trans. on Commun.}, vol.~69, no.~7, pp. 4781--4794, 2021.

\bibitem{xu2022proof}
T.~Xu, F.~Liu, C.~Masouros, and I.~Darwazeh, ``{Proof of Concept experiments of Joint Waveform Design for Integrated Sensing and Communications},'' in \emph{Proceedings of the 1st ACM MobiCom Workshop on Integrated Sensing and Communications Systems}, 2022, pp. 25--30.

\bibitem{ozkaptan2023software}
C.~D. Ozkaptan, H.~Zhu, E.~Ekici, and O.~Altintas, ``{Software-Defined MIMO OFDM Joint Radar-Communication Platform with Fully Digital mmWave Architecture},'' in \emph{2023 IEEE 3rd International Symposium on Joint Communications \& Sensing (JC\&S)}.\hskip 1em plus 0.5em minus 0.4em\relax IEEE, 2023, pp. 1--6.

\bibitem{liao2023faster}
Z.~Liao, F.~Liu, A.~Li, and C.~Masouros, ``{Faster-Than-Nyquist Symbol-Level Precoding for Wideband Integrated Sensing and Communications},'' \emph{arXiv preprint arXiv:2306.14509}, 2023.

\bibitem{9585492}
Z.~Zhang, Q.~Chang, F.~Liu, and S.~Yang, ``{Dual-Functional Radar-Communication Waveform Design: Interference Reduction Versus Exploitation},'' \emph{IEEE Commun. Lett.}, vol.~26, no.~1, pp. 148--152, 2022.

\bibitem{10065807}
M.~Wang and H.~Du, ``{Symbol-Level Precoding Design for Integrated Sensing and Communication},'' in \emph{2022 IEEE 8th International Conference on Computer and Communications (ICCC)}, 2022, pp. 967--971.

\bibitem{bjornson2013optimal}
E.~Bj{\"o}rnson, E.~Jorswieck \emph{et~al.}, ``Optimal resource allocation in coordinated multi-cell systems,'' \emph{Foundations and Trends{\textregistered} in Communications and Information Theory}, vol.~9, no. 2--3, pp. 113--381, 2013.

\bibitem{masouros2015exploiting}
C.~Masouros and G.~Zheng, ``{Exploiting Known Interference as Green Signal Power for Downlink Beamforming Optimization},'' \emph{IEEE Trans. on Signal processing}, vol.~63, no.~14, pp. 3628--3640, 2015.

\bibitem{li2020tutorial}
A.~Li, D.~Spano, J.~Krivochiza, S.~Domouchtsidis, C.~G. Tsinos, C.~Masouros, S.~Chatzinotas, Y.~Li, B.~Vucetic, and B.~Ottersten, ``{A Tutorial on Interference Exploitation via Symbol-Level Precoding: Overview, State-of-the-Art and Future Directions},'' \emph{IEEE Commun. Surveys \& Tutorials}, vol.~22, no.~2, pp. 796--839, 2020.

\bibitem{wei2023symbol}
Z.~Wei, R.~Xu, Z.~Feng, H.~Wu, N.~Zhang, W.~Jiang, and X.~Yang, ``Symbol-level integrated sensing and communication enabled multiple base stations cooperative sensing,'' \emph{IEEE Transactions on Vehicular Technology}, 2023.

\bibitem{wang2023resource}
X.~Wang, H.~Wu, Y.~Xu, H.~Cao, N.~Kumar, and J.~J. Rodrigues, ``{Resource Allocation in Multi-Cell Integrated Sensing and Communication Systems: A DRL Approach},'' in \emph{ICC 2023-IEEE International Conference on Communications}.\hskip 1em plus 0.5em minus 0.4em\relax IEEE, 2023, pp. 3210--3215.

\bibitem{li2022beamforming}
R.~Li, Z.~Xiao, and Y.~Zeng, ``{Beamforming Towards Seamless Sensing Coverage for Cellular Integrated Sensing and Communication},'' in \emph{2022 IEEE International Conference on Communications Workshops (ICC Workshops)}.\hskip 1em plus 0.5em minus 0.4em\relax IEEE, 2022, pp. 492--497.

\bibitem{xu2023fundamental}
Y.~Xu, L.~Xie, D.~Xu, and S.~Song, ``{Fundamental Limits and Base Station Selection for Collaborative Sensing in Perceptive Mobile Networks},'' in \emph{2023 IEEE International Mediterranean Conference on Communications and Networking (MeditCom)}.\hskip 1em plus 0.5em minus 0.4em\relax IEEE, 2023, pp. 97--102.

\bibitem{xu2022integrated}
D.~Xu, A.~Khalili, X.~Yu, D.~W.~K. Ng, and R.~Schober, ``{Integrated Sensing and Communication in Distributed Antenna Networks},'' \emph{arXiv preprint arXiv:2210.14880}, 2022.

\bibitem{xu2023integrated}
D.~Xu, C.~Liu, S.~Song, and D.~W.~K. Ng, ``{Integrated Sensing and Communication in Coordinated Cellular Networks},'' \emph{arXiv preprint arXiv:2305.01213}, 2023.

\bibitem{xu2023joint}
Y.~Xu, D.~Xu, L.~Xie, and S.~Song, ``{Joint BS Selection, User Association, and Beamforming Design for Network Integrated Sensing and Communication},'' \emph{arXiv preprint arXiv:2305.05265}, 2023.

\bibitem{jiang2023collaborative}
W.~Jiang, Z.~Wei, F.~Liu, Z.~Feng, and P.~Zhang, ``{Collaborative Precoding Design for Adjacent Integrated Sensing and Communication Base Stations},'' \emph{arXiv preprint arXiv:2310.08246}, 2023.

\bibitem{liu2023distributed}
X.~Liu, H.~Zhang, K.~Long, A.~Nallanathan, and V.~C. Leung, ``{Distributed Unsupervised Learning for Interference Management in Integrated Sensing and Communication Systems},'' \emph{IEEE Transa. on Wireless Commun.}, 2023.

\bibitem{zhang2023joint}
J.~Zhang, Z.~Fei, X.~Wang, P.~Liu, J.~Huang, and Z.~Zheng, ``{Joint Resource Allocation and User Association for Multi-Cell Integrated Sensing and Communication Systems},'' \emph{EURASIP Journal on Wireless Communications and Networking}, vol. 2023, no.~1, p.~64, 2023.

\bibitem{babu2023multi}
N.~Babu and C.~Masouros, ``{Multi-cell Coordinated Joint Sensing and Communications},'' in \emph{2023 Asilomar Conference}.\hskip 1em plus 0.5em minus 0.4em\relax IEEE, 2023.

\bibitem{wei2019multi}
Z.~Wei, C.~Masouros, K.-K. Wong, and X.~Kang, ``{Multi-cell interference exploitation: Enhancing the power efficiency in cell coordination},'' \emph{IEEE Transactions on Wireless Communications}, vol.~19, no.~1, pp. 547--562, 2019.

\bibitem{benzaghta2023designing}
M.~Benzaghta, G.~Geraci, D.~Lopez-Perez, and A.~Valcarce, ``{Designing Cellular Networks for UAV Corridors via Bayesian Optimization},'' \emph{arXiv preprint arXiv:2308.05052}, 2023.

\bibitem{li2007range}
J.~Li, L.~Xu, P.~Stoica, K.~W. Forsythe, and D.~W. Bliss, ``{Range Compression and Waveform Optimization for MIMO Radar: A Cram{\'e}r--Rao Bound Based Study},'' \emph{IEEE Trans. on Signal Processing}, vol.~56, no.~1, pp. 218--232, 2007.

\bibitem{grant2014cvx}
M.~Grant and S.~Boyd, ``{CVX: Matlab Software for Disciplined Convex Programming, version 2.1},'' 2014.

\bibitem{chrisSymbol}
K.~L. Law and C.~Masouros, ``{Symbol Error Rate Minimization Precoding for Interference Exploitation},'' \emph{IEEE Transactions on Communications}, vol.~66, no.~11, pp. 5718--5731, 2018.

\bibitem{bezdek2003convergence}
J.~C. Bezdek and R.~J. Hathaway, ``{Convergence of Alternating Optimization},'' \emph{Neural, Parallel \& Scientific Computations}, vol.~11, no.~4, pp. 351--368, 2003.

\bibitem{jiang2020faster}
H.~Jiang, T.~Kathuria, Y.~T. Lee, S.~Padmanabhan, and Z.~Song, ``{A Faster Interior Point Method for Semidefinite Programming},'' in \emph{2020 IEEE 61st annual symposium on foundations of computer science (FOCS)}.\hskip 1em plus 0.5em minus 0.4em\relax IEEE, 2020, pp. 910--918.

\bibitem{li2021symbol}
A.~Li, F.~Liu, X.~Liao, Y.~Shen, and C.~Masouros, ``Symbol-level precoding made practical for multi-level modulations via block-level rescaling,'' in \emph{2021 IEEE 22nd International Workshop on Signal Processing Advances in Wireless Communications (SPAWC)}.\hskip 1em plus 0.5em minus 0.4em\relax IEEE, 2021, pp. 71--75.

\end{thebibliography}

\end{document}